\documentclass[11pt]{article}

\usepackage[utf8]{inputenc} 
\usepackage[T1]{fontenc}    
\usepackage{hyperref}       
\usepackage{url}            
\usepackage{booktabs}       
\usepackage{amsfonts}       
\usepackage{nicefrac}       
\usepackage{microtype}      
\usepackage{amsmath,amsthm,amssymb}
\usepackage{bbm}
\usepackage{fullpage}
\usepackage{graphicx}
\usepackage{subcaption}

\newcommand{\reals}{\mathbb{R}}

\newtheorem{theorem}{Theorem}
\newtheorem{lemma}{Lemma}

\newtheorem{definition}{Definition}

\newcommand{\Val}{\text{Val}}

\begin{document}
\pagenumbering{gobble}

\title{Buying Data Over Time: Approximately Optimal Strategies for Dynamic Data-Driven Decisions}
\author{
Nicole Immorlica\thanks{Microsoft Research, \url{nicimm@microsoft.com}} 
\and
Ian A. Kash\thanks{Department of Computer Science, University of Illinois at Chicago, \url{iankash@uic.edu}}
\and
Brendan Lucier\thanks{Microsoft Research, \url{brlucier@microsoft.com}}
}

\date{}

\maketitle

\begin{abstract}
We consider a model where an agent has a repeated decision to make and wishes to maximize their total payoff. Payoffs are influenced by an action taken by the agent, but also an unknown state of the world that evolves over time. Before choosing an action each round, the agent can purchase noisy samples about the state of the world.  The agent has a budget to spend on these samples, and has flexibility in deciding how to spread that budget across rounds. We investigate the problem of choosing a sampling 
algorithm
that optimizes total expected payoff. For example: is it better to buy samples steadily over time, or to buy samples in batches? We solve for the optimal policy, and show that it is a natural instantiation of the latter. Under a more general model that includes per-round fixed costs, we prove that a variation on this batching policy is a $2$-approximation.
\end{abstract}

\clearpage
\pagenumbering{arabic}

\section{Introduction}

The growing demand for machine learning practitioners is a testament to the way data-driven decision making is shaping our economy.  Data has proven so important and valuable because so much about the current state of the world is \emph{a priori} unknown.  We can better understand the world by investing in data collection, but this investment can be costly; deciding how much data to acquire can be a non-trivial undertaking, especially in the face of budget constraints.
Furthermore, the value of data is typically not linear.  Machine learning algorithms often see diminishing returns to performance as their training dataset grows~\cite{Kalayeh1983,cortes1994learning}.
This non-linearity is further complicated by the fact that a data-driven decision approach is typically intended to replace some existing method, so its value is relative to the prior method's performance.

As a motivating example for these issues, consider a politician who wishes to accurately represent the opinion of her constituents.  These constituents have a position on a policy, say the allocation of funding to public parks.  The politician must choose her own position on the policy 
or abstain from the discussion.
If she states a position, she experiences a disutility that is increasing in the distance of her position from that of her constituents.  If she abstains, she incurs a fixed cost for failing to take a stance.  To help her make an optimal decision she can hire a polling firm that collects data on the participants' positions.

We focus on the dynamic element of this story.  In many decision problems, the state of the world evolves over time.  In the example above, the opinions of the constituents might change as time passes, impacting the optimal position of the politician.  As a result, data about the state of the world becomes stale.  Furthermore, many decisions are not made a single time; instead, decisions are made repeatedly.  In our example, the politician can update funding levels each fiscal quarter. 

When faced with budget constraints on data collection and the issue of data staleness, decisions need to be made about when to collect data and when to save budget for the future, and whether to make decisions based on stale data or apply a default, non-data-driven policy.  Our main contribution is a framework that models the impact of such budget constraints on data collection strategies.  In our example, the politician has a budget for data collection.  A polling firm charges a fixed cost to initiate a poll (e.g., create the survey) plus a fee per surveyed participant.  The politician may not have enough budget to hire the firm to survey every constituent every quarter.  Should she then survey fewer constituents every quarter?  Or survey a larger number of constituents every other quarter, counting on the fact that opinions do not drift too rapidly?

We initiate the study with arguably the simplest model that exhibits this tension.  The state of the world (constituents' opinions) is hidden but drawn from a known prior distribution, then evolves stochastically.  Each round, the decision-maker (politician) can collect one or more noisy samples that are correlated with the hidden state at a cost affine in the number of samples (conduct a poll). Then she chooses an action and incurs a loss.  Should the decision-maker not exhaust her budget in a given round, she can bank it for future rounds.  A sampling algorithm describes an online policy for scheduling the collection of samples given the budget and past observations.

We instantiate this general framework by assuming Gaussian prior, perturbations and sample noise.\footnote{A Gaussian prior is justified in our running example if we assume a large population limit of constituents' opinions.  That the prior estimate of drift is also Gaussian is likewise motivated as the number of periods grows large.  We discuss alternative distributional assumptions on the prior, perturbations and noise in Section~\ref{sec:extension}.}  We capture the decisions that need to be made as the problem of estimating the current state value, using the classic squared loss to capture the cost of making a decision using imprecise information.  Alternatively, there is always the option to not make a decision based on the data and instead accept a default constant loss.  We assume a budget on the number of samples collected per unit time, and importantly this budget can be banked for future rounds if desired.

\subsection{A Simple Example.} To illustrate our technical model, 
suppose the hidden state (constituents' average opinion) is initially drawn from a mean-zero Gaussian of variance $1$.  In each round, the state is subject to mean-zero Gaussian noise of variance $1$ (the constituents update their opinions), which is added to the previous round's state.  Also, any samples we choose to take are also subject to mean-zero Gaussian noise of variance $1$ (polls are imperfect).  Our budget for samples is $1$ per period, and one can either guess at the hidden state (incurring a penalty equal to the squared loss) or pass and take a default loss of $3/4$.  What is the expected average loss of the policy that takes a single sample each round, and then takes the optimal action?  As it turns out, the expected loss is precisely $\phi - 1 \approx 0.618$, where $\phi$ is the golden ratio $\frac{1+\sqrt{5}}{2}$ (see Section~\ref{sec:example.simple} for the analysis).  
However, this is not optimal: 
saving up the allotted budget and taking two samples every other round leads to an expected loss of $\frac{0.75 + \sqrt{2}-1}{2} \approx 0.582$.  The intuition behind the improvement is that taking a single sample every round beats the outside option, but not by much; it is better to beat the outside option significantly on even-numbered rounds (by taking 2 samples), then simply use the outside option on odd-numbered rounds.
It turns out that one cannot improve on this by saving up for 3 or more rounds to take even more samples all at once.  However, one can do better by alternating between taking no samples for two periods and then two samples each for two periods, which results in a long-run average loss of $\approx 0.576$.

\subsection{Our Results.}
As we can see from the example above, the space of policies to consider
is
quite large.
One simple observation is that since samples become stale over time it is never optimal to collect samples and then take the outside option (i.e., default fixed-cost action) in the same round;
it would be better to defer 
data collection to later rounds where decisions will be made based on data.
As a result, a natural class of policies to consider is those which alternate between collecting samples and saving budget.  
Such ``on-off'' policies can be thought of as engaging in ``data drives'' while neglecting data collection the rest of the time.

Our main result is that these on-off policies are asymptotically optimal, with respect to all dynamic policies.  Moreover, it suffices to collect samples at a constant rate during the sampling part of the policy's period.  Our argument is constructive, and we show how to compute an asymptotically optimal policy.  This policy divides time into exponentially-growing chunks and collects data in the latter end of each chunk.  

The solution above assumes that costs are linear in the number of samples collected.  We next consider a more general model with a fixed up-front cost for the first sample collected in each round.  This captures the costs associated with setting up the infrastructure to collects samples on a given round, such as 
hiring a polling firm which uses a two-part tariff.
Under such per-round costs, it can be suboptimal to sample in sequential periods (as in an on-off policy), as this requires paying the fixed cost twice. For this generalized cost model, we consider simple and approximately optimal policies.  When evaluating performance, we compare against a null ``baseline'' policy that eschews data collection and simply takes the outside option every period.  We define the value of a policy to be its improvement over this baseline, so that the null policy has a value of $0$ and every policy has non-negative value.  While this is equivalent to simply comparing the expected costs of policies this alternative measure is intended to capture how well a policy leverages the extra value obtainable from data; we feel that this more accurately reflects the relative performance of different policies.  

We focus on a class of {\em lazy policies} that collect samples only at times when the variance of the current estimate is worse than the outside option. This class captures a heuristic based on a threshold rule: the decision-maker chooses to collect data when they do not have enough information to gain over the outside option. We show the optimal lazy policy is a $1/2$-approximation to the optimal policy.  The result is constructive, and we show how to compute an asymptotically optimal lazy policy.  Moreover, this approximation factor is tight for lazy policies.

To derive these results, we begin with the well-known fact that the expected loss under the squared loss cost function is the variance of the posterior.  We use an analysis based on Kalman filters~\cite{kalman1960new}, which are used to solve localization problems in domains such as astronautics~\cite{lefferts1982kalman}, robotics~\cite{thrun2000probabilistic}, and traffic monitoring~\cite{work2008ensemble}, to characterize the evolution of variance given a sampling policy. We show how to maximize value using geometric arguments and local manipulations to transform an optimal policy into either an on-off policy or a lazy policy, respectively.

We conclude with two extensions.  We described our results for a discrete-time model, but one might instead consider a continuous-time variant in which samples, actions, and state evolution occur continuously.  We show how to extend all of our results to such a continuous setting.  Second, we describe a non-Gaussian instance of our framework, where the state of the world is binary and switches with some small probability each round.  We solve for the optimal policy, and show that (like the Gaussian model) it is characterized by non-uniform, bursty sampling.

\subsection{Other Motivating Examples.}  We motivated our framework with a toy example of a politician polling his or her constituents.  But we note that the model is general and applies to other scenarios as well.  For example, suppose a phone uses its GPS to collect samples, each of which provides a noisy estimate of location (reasonably approximated by Gaussian noise).  The ``cost'' of collecting samples is energy consumption, and the budget constraint is that the GPS can only reasonably use a limited portion of the phone's battery capacity.  The worse the location estimate is, the less useful this information is to apps; sufficiently poor estimates might even have negative value.  However, as an alternative, apps always have the outside option of providing location-unaware functionality.  Our analysis shows that it is approximately optimal to extrapolate from existing data to estimate the user’s location most of the time, and only use the GPS in “bursts” once the noise of the estimate exceeds a certain threshold.  Note that in this scenario the app never observes the ``ground truth'' of the phone's location.  Similarly, our model might capture the problem faced by a firm that runs user studies when deciding which features to include in a product, given that such user studies are expensive to run and preferences may shift within the population of customers over time.

\subsection{Future Directions.}
Our results provide insight into the trade-offs involved in designing data collection policies in dynamic settings.  We construct policies that navigate the trade-off between cost of data collection and freshness of data, and show how to optimize data collection schedules in a setting with Gaussian noise.  But perhaps our biggest contribution is conceptual, in providing a framework in which these questions can be formalized and studied.  We view this work as a first step toward a broader study of the dynamic value of data. An important direction for future work is to consider other models of state evolution and/or sampling within our framework, aimed at capturing other applications. For example, if the state evolves in a heavy-tailed manner, as in the non-Gaussian instance explored in Section~\ref{sec:extension}, then we show it is beneficial to take samples regularly in order to detect large, infrequent jumps in state value, and then adaptively take many samples when such a jump is evident.  We solve this extension only for a simple two-state Markov chain.  Can we quantify the dynamic value of data and find an (approximately) optimal and simple data collection policy in a general Markov chain?

\subsection{Related work}

While we are not aware of other work addressing the value of data in a dynamic setting, there has been considerable attention paid to the value of data in static settings.  Arietta-Ibarra et al.~\cite{arrieta2018should} argue that the data produced by internet users is so valuable that they should be compensated for their labor.  Similarly, there is growing appreciation for the value of the data produced on crowdsourcing platforms like Amazon Mechanical Turk~\cite{buhrmester2011amazon,ipeirotis2010analyzing}.  Other work has emphasized that not all crowdsourced data is created equal and studied the way tasks and incentives can be designed to improve the quality of information gathered~\cite{fothergill2012instructing,shah2015double}.  Similarly, data can have non-linear value if individual pieces are substitutes or complements~\cite{chen2016informational}.  Prediction markets can be used to gather information over time, with participants controlling the order in which information is revealed~\cite{cowgill2009using}.

There is a growing line of work attempting to determine the marginal value of training data for deep learning methods.  Examples include training data for classifying medical images~\cite{Cho2015} and chemical processes~\cite{Beleites2013SampleSP}, as well as for more general problems such as estimating a Gaussian distribution~\cite{Kalayeh1983}.  These studies consider the static problem of learning from samples, and generally find that additional training data exhibits decreasing marginal value.  Koh and Liang~\cite{koh2017understanding} introduced the use of influence functions to quantify how the performance of a model depends on individual training examples.

While we assume samples are of uniform quality, other work has studied agents who have data of different quality or cost~\cite{liang2017dynamic,chen2018optimal,fang2007putting}.
Another line studies the way that data is sold in current marketplaces~\cite{stahl2014data}, as well as proposing new market designs~\cite{li2012pricing}.  This includes going beyond markets for raw data to markets which acquire and combine the outputs of machine learning models~\cite{storkey2011machine}.

Our work is also related to statistical and algorithmic aspects of learning a distribution from samples.  A significant body of recent work has considered problems of learning Gaussians using a minimal number of noisy and/or adversarial samples~\cite{KMV10,DKKLMS16,DKS17,LRV16,Diakonikolas2018}.
In comparison, we are likewise interested in learning a hidden Gaussian from which we obtain noisy samples (as a step toward determining an optimal action), but instead of robustness to adversarial noise we are instead concerned about optimizing the split of samples across
time periods in a purely stochastic setting.


Our investigation of data staleness is closely related to the issue of concept drift in streaming algorithms; see, e.g., Chapter 3 of~\cite{Aggarwal2006}
Concept drift refers to scenarios where the data being fed to an algorithm is pulled from a model that evolves over time, so that, for example, a solution built using historical data will eventually lose accuracy.  Such scenarios arise in problems of histogram maintenance~\cite{Gilbert2002}, dynamic clustering~\cite{Aggarwal2003Clustering}, and others.  One problem is to quantify the amount of drift occurring in a given data stream~\cite{Aggarwal2003}.  Given that such drift is present, one approach to handling concept drift is via sliding-window methods, which limit dependence on old data~\cite{Datar2002}.  The choice of window size captures a tension between using a lot of stale data or a smaller amount of fresh data.  However,
in work on concept drift one typically cannot control the rate at which data is collected.

Another concept related to staleness is the ``age of information.''  This captures scenarios where a source generates frequent updates and a receiver wishes to keep track of the current state, but due to congestion 
in the transmission 
technology (such as a queue or database locks) it is optimal to limit the rate at which updates are sent~\cite{kaul2012real,song1990performance}.  
Minimizing the age of information can be captured as a limit of our model where a single sample suffices to provide perfect information.  Recent work has examined variants of the model where generating updates is costly~\cite{hao2020regulating}, but the focus in this literature is more on the management of the congestible resource.  Closer to our work, several recent papers have eliminated the congestible resource and studied issues such as an energy budget that is stochastic and has limited storage capacity~\cite{wu2017optimal} and pricing schemes for when sampling costs are non-uniform~\cite{wang2019dynamic,zhang2019price}.  Relative to our work these papers have simpler models of the value of data and focus on features of the sampling policy given the energy technology and pricing scheme, respectively.

\section{Model}
\label{sec:model.discrete}


We first describe our general framework, then describe a specific instantiation of interest in Section~\ref{sec:model.gaussian}.
Time occurs in rounds, indexed by $t = 1, 2, \dotsc$.  There is a hidden state variable $x_t \in \Omega$ that evolves over time according to a stochastic process.  The initial state $x_1$ is drawn from known distribution $F_1$.  Write $m_t$ for the (possibly randomized) evolution mapping applied at round $t$, so that $x_{t+1} \leftarrow m_t(x_t)$.  

In every round, the decision-maker chooses an action $y_t \in A$, and then suffers a loss $\ell(y_t,x_t)$ that depends on both the action and the hidden state.  The evolution functions $(m_t)$ and loss function $\ell$ are known to the decision-maker, but neither the state $x_t$ nor the loss $\ell(y_t,x_t)$ is directly observed.\footnote{Assuming that the ground truth for $\ell(y_t,x_t)$ is unobserved captures scenarios like our political example, and approximates settings where the decision maker only gets weak feedback, feedback at a delay, or feedback in aggregate over a long period of time.  Observing the loss provides additional information about $x_{t+1}$, and this could be considered a variant of our model where the decision-maker gets some number of samples ``for free'' each round from observing a noisy version of the loss.}  Rather, on each round before choosing an action, the decision-maker can request one or more independent samples that are correlated with $x_t$, drawn from a known distribution $\Gamma(x_t)$.  

Samples are costly, and the decision-maker has a budget that can be used to obtain samples.  The budget is $B$ per round, and can be banked across rounds.  A sampling policy results in a number of samples $s_t$ taken in each round $t$, which can depend on all previous observations.  The cost of taking $s_t$ samples in round $t$ is $C(s_t) \geq 0$.  We assume that $C$ is non-decreasing and $C(0) = 0$.  A sampling policy is \emph{valid} if $\sum_{t = 1}^T C(s_t) \leq B\cdot T$ for all $T$.  For example, $C(s_t) = s_t$ corresponds to a cost of $1$ per sample, and setting $C(s_t) = s_t + z \cdot \mathbbm{1}_{s_t > 0}$ adds an additional cost of $z$ for each round in which at least one sample is collected. 

To summarize: on each round, the decision-maker chooses a number of samples $s_t$ to observe, then chooses an action $y_t$.  Their loss $\ell(y_t,x_t)$ is then realized, the value of $x_t$ is updated to $x_{t+1}$, and the process proceeds with the next round.  
The goal is to minimize the expected long-run average of $\ell(y_t,x_t)$, in the limit as $t \to \infty$,
subject to $\sum_{t=0}^{T} C(s_t) \leq B\cdot T$ for all $T \geq 1$.  

\subsection{Estimation under Gaussian Drift}\label{sec:model.gaussian}


We will be primarily interested in the following instantiation of our general framework.  The hidden state variable is a real number (i.e., $\Omega = \reals$) and the decision-maker's goal is to estimate the hidden state in each round.  The initial state is $x_1 \sim N(0,\rho)$, a Gaussian with mean $0$ and variance $\rho > 0$.  Moreover, the evolution process $m_t$ sets $x_{t+1} = x_t + \delta_t$, where each $\delta_t\sim N(0,\rho)$ independently.  We recall that the decision-maker knows the evolution process (and hence $\rho$) but does not directly observe the realizations $\delta_t$.

Each sample in round $t$ is drawn from $N(x_t, \sigma)$ where $\sigma > 0$.  
Some of our results will also allow fractional sampling, where we think of an $\alpha \in (0,1)$ fraction of a sample as a sample drawn from $N(x_t, \sigma/\alpha)$.\footnote{One can view fractional sampling as modeling scenarios where the value of any one single sample is quite small; i.e., has high variance, so that a single ``unit'' of variance is derived from taking many samples.  E.g., sampling a single constituent in our polling example.  It also captures settings where it is possible to obtain samples of varying quality with different levels of investment.}
%
The action space is 
$A = \reals \cup \{\perp\}$.  If the decision-maker chooses $y_t \in \reals$, her loss is the squared error of her estimate $(y_t-x_t)^2$.  If she is too unsure of the state, she may instead take a default action $y_t = \perp$, which corresponds to not making a guess; this results in a constant loss of $c > 0$.  
%
Let $G_t$ be a random variable whose law is the decision maker's posterior after observing whatever samples are taken in round $t$ as well as all previous samples.  The decision maker's subjective expected loss when guessing $y_t \in \reals$ is $E[(y_t - G_t)^2]$.  This is well known to be minimized by taking $y_t = E[G_t]$, and that furthermore the expected loss is $E[(E[G_t] - G_t)^2] = Var(G_t)$.  It is therefore optimal to guess $y_t = E[G_t]$ if and only if Var$(G_t) < c$, otherwise pass.  

We focus on deriving approximately optimal sampling algorithms.  To do so, we need to track the variance of $G_t$ as a function of the sampling strategy.  
As the sample noise and random state permutations are all zero-mean Gaussians,
$G_t$ is a zero-mean Gaussian as well, and the evolution of its variance has a simple form.

\begin{lemma}\label{lem:kalman}
Let $v_t$ be the variance of $G_t$ and suppose each $\delta_t \sim N(0,\rho)$ independently, and that each sample is subject to zero-mean Gaussian noise with variance $\sigma$. Then, if the decision-maker takes $s$ samples in round $t+1$, the variance of $G_{t+1}$ is 
\[v_{t+1}=\frac{v_t+\rho}{1 + \frac{s}{\sigma}(v_t+\rho)}.\]
\end{lemma}

The proof, which is deferred to the appendix along with all other proofs,
follows from our model being a special case of the model underlying a Kalman filter.

The optimization problem therefore reduces to choosing a number of samples $s_t$ to take in each round $t$ in order to minimize the long-run average of $\min(v_t,c)$, the loss of the optimal action.  That is, the goal is to minimize 
$\lim\sup_{T \to \infty} \frac{1}{T} \sum_{t = 1}^{T} \min(v_t,c),$
where we take the superior limit so that the quantity is defined even when the average is not convergent.  We choose $C(s_t) = s_t + z \cdot \mathbbm{1}_{s_t > 0}$, so this optimization is subject to the budget constraint that, at each time $T \geq 1$, $\sum_{t=1}^Ts_t + z \cdot \mathbbm{1}_{s_t > 0} \leq BT$.  This captures two kinds of information acquisition costs faced by the decision-maker.  First she faces a cost per sample, which we have normalized to one.  Second, she faces a fixed cost $z$ (which may be 0) on each day she chooses to take samples, expressed in terms of the number of samples that could instead have been taken on some other day had this cost not been paid.  This captures the costs associated with setting up the infrastructure to collects samples on a given round, such as getting data collectors to the location where they are needed, hiring a polling firm which uses a two-part tariff, or establishing a satellite connection to begin using a phone's GPS.

A useful baseline performance is the cost of a policy that
 takes no samples and simply chooses the outside option at all times.  We refer to this as the \emph{null policy}.  The \emph{value} of a sampling policy $s$, denoted $\Val(s)$, is defined to be the difference between its cost and the cost of the null policy:
$\lim\inf_{T \to \infty}\frac{1}{T} \sum_{t=1}^T \max(c-v_t, 0).$
Note that maximizing value is
equivalent to minimizing cost, which we illustrate in Section~\ref{sec:viz}.
We say that a policy is $\alpha$-approximate if its value is at least an $\alpha$ fraction of the optimal policy's value.  

\section{Analyzing Variance Evolution}
\label{sec:piecewise}

Before moving on to our main results, we show how to analyze the evolution of the variance resulting from a given sampling policy.  We first illustrate our model with a particularly simple class of policies: those where $s_t$ takes on only two possible values.  We then analyze arbitrary periodic policies, and show via contraction that they result in convergence to a periodic variance evolution.

\subsection{Visualizing the Decision Problem}
\label{sec:viz}

To visualize the problem, we begin by plotting the result of an example policy where the spending rate is constant for some interval of rounds, then shifts to a different constant spending rate.
Figure~\ref{fig:pc-policy} illustrates one such policy.  The spending rates are indicated as alternating line segments, while the variance is an oscillating curve, always converging toward the current spending rate.  Note that this particular policy is periodic, in the sense that the final variance is the same as the initial variance.  The horizontal line gives one possible value for the cost of the outside option.  Given this, the optimal policy is to guess whenever the orange curve is below the green line and take the outside option whenever it is above it.  Thus, the loss associated with this spending policy is given by the orange shaded area in 
Figure~\ref{fig:pc-policy}.
Minimizing this loss is equivalent to maximizing the green shaded area, which corresponds to the value of the spending policy.  The null policy, which takes no samples and has variance greater than $c$ always (possibly after an initial period if $v_0 < c$), has value $0$.

\begin{figure*}[t]
\centering
\includegraphics[width=0.4\textwidth]{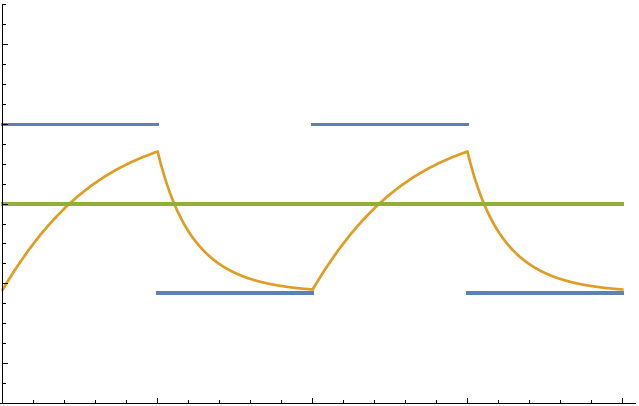}
\includegraphics[width=0.4\textwidth]{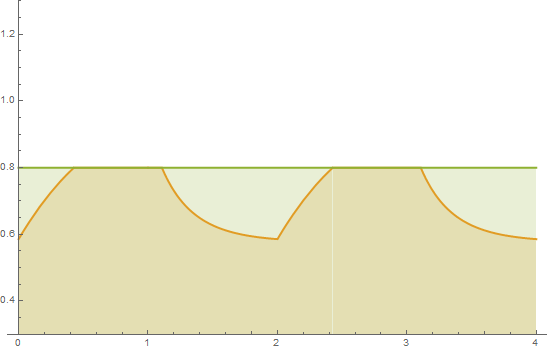}
\caption{The variance for a piecewise-constant sampling policy, and its loss and benefit}
\label{fig:pc-policy}
\end{figure*}


\subsection{Periodic Policies}
\label{sec:model.periodic}

We next consider policies that are periodic.  A \emph{periodic policy} with period $R$ has the property that $s_t = s_{t+R}$
for all $t \geq 1$.  
Such policies are natural and have useful structure.
In a periodic policy, the variance $(v_t)$ converges uniformly to being periodic in the limit as $t \to \infty$.  This follows because the impact of sampling on variance is a contraction map.

\begin{definition}
Given a normed space $X$ with norm $||\cdot||$, a mapping $\Psi \colon X \to X$ is a \emph{contraction map} if there exists a $k < 1$ such that, for all $x,y \in X$, $||\Psi(x) - \Psi(y)|| \leq k ||x - y||$.
\end{definition}

\begin{lemma}
\label{lem.contraction}
Fix a sampling policy $s$, and a time $R \geq 1$, 
and suppose that $s$ takes a strictly positive number of samples in each round $t \leq R$.
Let $\Psi$ be the mapping defined as follows: supposing that $v_0 = x$ and $v$ is the variance function resulting from sampling policy $s$, set $\Psi(x) := v_R$.  Then $\Psi$ is a contraction map over the non-negative reals, under the absolute value norm.
\end{lemma}
The proof appears in 
Appendix~\ref{app:opt}.
It is well known that a contraction map has a unique fixed point, and repeated application will converge to that fixed point.  Since we can view the impact of the periodic sampling policy as repeated application of mapping $\Psi$ to the initial variance in order to obtain $v_0, v_R, v_{2R}, \dotsc$, we conclude that the variance will converge uniformly to a periodic function for which $v_t = v_{t+R}$.  Thus, for the purpose of evaluating long-run average cost, it will be convenient (and equivalent) to replace the initial condition on $v_0$ with a periodic boundary condition $v_0 = v_R$, and then choose $s$ to minimize the average cost over a single period,
$\frac{1}{R} \int_0^R \min\{v_t, c\}dt,$
subject to the budget constraint that, at any round $T \leq R$, we have
$\sum_{t=1}^{T}s_t \leq BT$.

\subsection{Lazy Policies}

Write $\tilde{v} = v_{t-1} + \rho$ for the variance that would be obtained in round $t$ if $s_t = 0$.
We say that a policy is \emph{lazy} if $s_t = 0$ whenever $\tilde{v}_t < c$.
%
That is, samples are collected only at times where the variance would otherwise be at or above the outside option value $c$.  
Intuitively, we can think of such a policy as collecting a batch of samples in one round, then ``free-riding'' off of the resulting information in subsequent rounds.  The free-riding occurs until the posterior variance grows large enough that it becomes better to select the outside option, at which point the policy may collect another batch of samples.

If a policy is lazy, then its variance function $v$ 
increases by $\rho$ whenever $\tilde{v}_t < c$, with downward steps only at times corresponding to when samples are taken.
Furthermore, the value of such a policy decomposes among these sampling instances: 
for any $t$ where $s_t > 0$, 
resulting in a variance of $v_t < c$, if we write $h = \lfloor c - v_t \rfloor$ then we can attribute a value of 
$\frac{1}{2} h (h+1) + (h+1)(c - v_t - h)$.  
Geometrically, this is the area of the ``discrete-step triangle'' formed between the increasing sequence of variances $v_t$ and the constant line at $c$, over the time steps $t, \dotsc, t+h+1$.

%
%

\subsection{On-Off Policies}

An On-Off policy is a periodic policy parameterized by a time interval $T$ and a sampling rate $S$.  Roughly speaking, the policy alternates between intervals where it samples at a rate of $S$ each round, and intervals where it does not sample.  The two interval lengths sum to $T$, and the length of the sampling interval is set as large as possible subject to budget feasibility.  More formally,
the policy sets $s_t =0$ for all $t \leq (1-\alpha) \cdot T$, where $\alpha = \min\{B / S, 1\} \in [0,1]$ and $s_t = S$ for all $t$ such that $(1-\alpha)T < t \leq T$.  
This policy is then repeated, on a cycle of length $T$.  The fraction $\alpha$ is chosen to be as large as possible, subject to the budget constraint.

\subsection{Simple Example Revisited}
\label{sec:example.simple}
We can now justify the simple example we presented in the introduction, where $\rho = \sigma = 1$, $B = 1$, and $c = 0.75$.  The policy that takes a single sample each round is periodic with period $1$, and hence will converge to a variance that is likewise equal each round.  This fixed point variance, $v^*$, satisfies $v^* = \frac{v^*+1}{1+(v^*+1)}$ by Lemma~\ref{lem:kalman}.  Solving for $v^*$ yields $v^* = \frac{\sqrt{5}-1}{2} < 0.75$, which is the average cost per round.

If instead the policy takes $k$ samples every $k$ rounds, this results in a variance that is periodic of period $k$.  After the round in which samples are taken, the fixed-point variance satisfies $v^* = \frac{v^*+k}{1+k(v^*+k)}$, again by Lemma~\ref{lem:kalman}.  Solving for $v^*$, and noting that $v^*+1 \geq 1 > c$, yields that the cost incurred by this policy is minimized when $k = 2$.

To solve for the policy that alternates between taking no samples for two round, followed by taking two samples on each of two rounds, suppose the long-run, periodic variances are $v_1, v_2, v_3, v_4$, where samples are taken on rounds $3$ and $4$.  Then we have $v_2 = v_1+1$, $v_3 = \frac{v_2+1}{1+2(v_2+1)}$, $v_4 = \frac{v_3+1}{1+2(v_3+1)}$, and $v_1 = v_4 + 1$.  Combining this sequence of equations yields $4v_1^2 + 4v_1 - 13 = 0$, which we can solve to find $v_1 = \frac{-1 + \sqrt{14}}{2} \approx 1.3708$.  Plugging this into the equations for $v_2, v_3, v_4$ and taking the average of $\min\{v_i, 0.75\}$ over $i \in \{1,2,3,4\}$ yields the reported average cost of $\approx 0.576$.

\section{Solving for the Optimal Policy}
\label{sec:optimal}

In this section we show that when the cost of sampling is linear in the total number of samples taken (i.e., $z = 0$)\footnote{Recall that $z$ is the fixed per-round cost of taking a positive number of samples.  Even when $z = 0$, there is still a positive per-sample cost.}, and when fractional sampling is allowed, then the supremum value over all on-off policies is an upper bound on the value of any policy.  This supremum is achieved in the limit as the time interval $T$ grows large.  So, while no individual policy achieves the supremum, one can get arbitrarily close with an on-off policy of sufficiently long period.  Proofs appear in 
Appendix~\ref{app:opt}.

We begin with some definitions.  For a given period length $T > 0$, write $s^T$ for the on-off policy of period $T$ with optimal long-run average value.  Recall $\Val(s^T)$ is the value of policy $s_T$.  We first argue that larger time horizons lead to better on-off policies.

\begin{lemma}\label{lem:onoff.monotone}
With fractional samples, for all $T > T'$, we have $\Val(s^T) > \Val(s^{T'})$.
\end{lemma}

Write $V^* = \sup_{T \to \infty} \Val(s^T)$.  Lemma~\ref{lem:onoff.monotone} implies that $V^* = \lim_{T \to \infty} \Val(s^T)$ as well.  
We show that no policy satisfying the budget constraint can achieve value greater than $V^*$.

\begin{theorem}
\label{thm:opt}
With fractional samples, the value of any valid policy $s$ is at most $V^*$.
\end{theorem}

The proof of Theorem~\ref{thm:opt} proceeds in two steps.  First, for any given time horizon $T$, it is suboptimal to move from having variance below the outside option to above the outside option; one should always save up budget over the initial rounds, then keep the variance below $c$ from that point onward.  This follows because the marginal sample cost of reducing variance diminishes as variance grows, so it is more sample-efficient to recover from very high variance once than to recover from moderately high variance multiple times.  

Second, one must show that it is asymptotically optimal to keep the variance not just below $c$, but uniform.  This is done by a potential argument, illustrating that a sequence of moves aimed at ``smoothing out'' the sampling rate can only increase value and must terminate at a uniform policy.  The difficulty is that a sample affects not only the value in the round it is taken, but in all subsequent rounds.  We make use of an amortization argument that appropriately credits value to samples, and use this to construct the sequence of adjustments 
that
increase overall value while bringing the sampling sequence closer to uniform in an appropriate metric.

We also note that it is straightforward to compute the optimal on-off policy 
for a given time horizon $T$,
by choosing the sampling rate that maximizes [value per round] $\times$ [fraction of time the policy is ``on''].  One can implement a policy whose value asymptotically approaches $V^*$ by repeated doubling of the time horizon.  Alternatively, since $\lim_{T \to \infty} \Val(s^T) = V^*$, $s^T$ will be an approximately optimal policy for sufficiently large $T$.

\section{Approximate Optimality of Lazy Policies}
\label{sec:lazy_approx}

In the previous section we solved for the optimal policy when $z = 0$, meaning that there is no fixed per-round cost when sampling.  We now show that for general $z$, lazy policies are approximately optimal, obtaining at least $1/2$ of the value of the optimal policy.  
All proofs are deferred to Appendix~\ref{sec:appendix.approx}.


We begin with a lemma that states that, for any valid sampling policy and any sequence of timesteps, it is possible to match the variance at those timesteps with a policy that only samples at precisely those timesteps, and the resulting policy will be valid.  

\begin{lemma}\label{lem:atoms}
Fix any valid sampling policy $s$ (not necessarily lazy) with resulting variances $(v_t)$, and any sequence of timesteps $t_1 < t_2 < \dotsc < t_\ell < \dotsc$.  Then there is a valid policy $s'$ such that $\{t ~|~ s'_t  > 0 \} \subseteq \{t_1, \dotsc, t_\ell, \dotsc\}$, resulting in a variances $(\breve{v}_t)$ with $\breve{v}_{t_i} \leq v_{t_i}$ for all $i$.
\end{lemma}


The intuition is that if we take all the samples we would have spent between timesteps $t_\ell$ and $t_{\ell+1}$ and instead spend them all at $t_{\ell+1}$ the result will be a (weakly) lower variance at $t_{\ell+1}$.
We next show that any policy can be converted into a lazy policy at a loss of at most half of its value.

\begin{figure*}[t]
\centering
\includegraphics[width=0.7\textwidth]{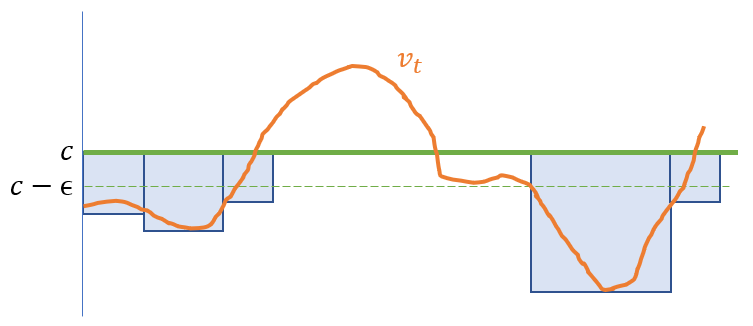}
\caption{Visualizing the construction in the proof of Theorem~\ref{thm:2approx.costs}.  Variance (vertical) is plotted against time (horizontal).  We approximate the value of an optimal policy's variance (orange) given $c$ (green).  
The squares (drawn in blue) cover the gap between the curves, except possibly when $|v_t - c| < \epsilon$ (for technical reasons).  The lazy policy samples on rounds corresponding to the left edge of each square, bringing the variance to each square's bottom-left corner.}
\label{fig:squares}
\end{figure*}

\begin{theorem}\label{thm:2approx.costs}
The optimal lazy policy is $1/2$-approximate.
\end{theorem}

See Figure~\ref{fig:squares} for an illustration of the intuition behind the result.  Consider an arbitrary policy $s$, with resulting variance sequence $(v_t)$.  Imagine covering the area between $(v_t)$ and $c$ with squares, drawn left to right with their upper faces lying on the outside option line, each chosen just large enough so that $v_t$ never falls below the area covered by the squares.
The area of the squares is an upper bound on $\Val(s)$.  Consider a lazy policy that drops a single atom on the left endpoint of each square, bringing the variance to the square's lower-left corner.  The value of this policy covers at least half of each square.  Moreover, Lemma~\ref{lem:atoms} implies this policy is (approximately) valid, as it matches variances from the original policy, possibly shifted early by a constant number of rounds.  This shifting can introduce non-validity; we fix this by delaying the policy's start by a constant number of rounds without affecting the asymptotic behavior.

The factor $1/2$ in Theorem~\ref{thm:2approx.costs} is tight.  To see this, fix the value of $c$ and allow the budget $B$ to grow arbitrarily large.  Then the optimal value tends to $c$ as the budget grows, since the achievable variance on all rounds tends to $0$.  However, the lazy policy cannot achieve value greater than $c/2$, as this is what would be obtained if the variance reached $0$ on the rounds on which samples are taken.

Finally, while this result is non-constructive, one can compute a policy whose value approaches an upper bound on the optimal lazy policy, in a similar manner to the optimal on-off policy.  One can show the best lazy policy over any finite horizon has an ``off'' period (with no sampling) followed by an ``on'' period (where $v_t \leq c$).  One can then solve for the optimal number of samples to take whenever $\tilde{v}_t > c$ by optimizing either value per unit of (fixed plus per-sample) sampling cost, or by fully exhausting the budget, whichever is better.  
See Lemma~\ref{lem:solve-optimal-atomic-discrete.costs} in the appendix for details.

\section{Extensions and Future Directions}
\label{sec:extension}


We describe two extensions of our model in the appendix.  First, we consider a continuous-time variant where samples can be taken continously
subject to a flow cost, in addition to being requested as discrete atoms.  The decision-maker selects actions continuously, and aims to minimize loss over time.  All of our results carry forward to this continuous extension.  



Second, returning to discrete time, we consider a non-Gaussian instance of our framework.  
In this model, there is a binary hidden state of the world, which flips each round independently with some small probability $\epsilon > 0$.  The decision-maker's action in each round is to guess the hidden state of this simple two-state Markov process, and the objective is to maximize the fraction of time that this guess is made correctly.  Each sample is a binary signal correlated with the hidden state, matching the state of the world with probability $\tfrac{1}{2} + \delta$ where $\delta > 0$.  The decision-maker can adaptively request samples in each round, subject to the accumulating budget constraint, before making a guess.

\begin{figure*}[t]
\centering
\includegraphics[width=\textwidth]{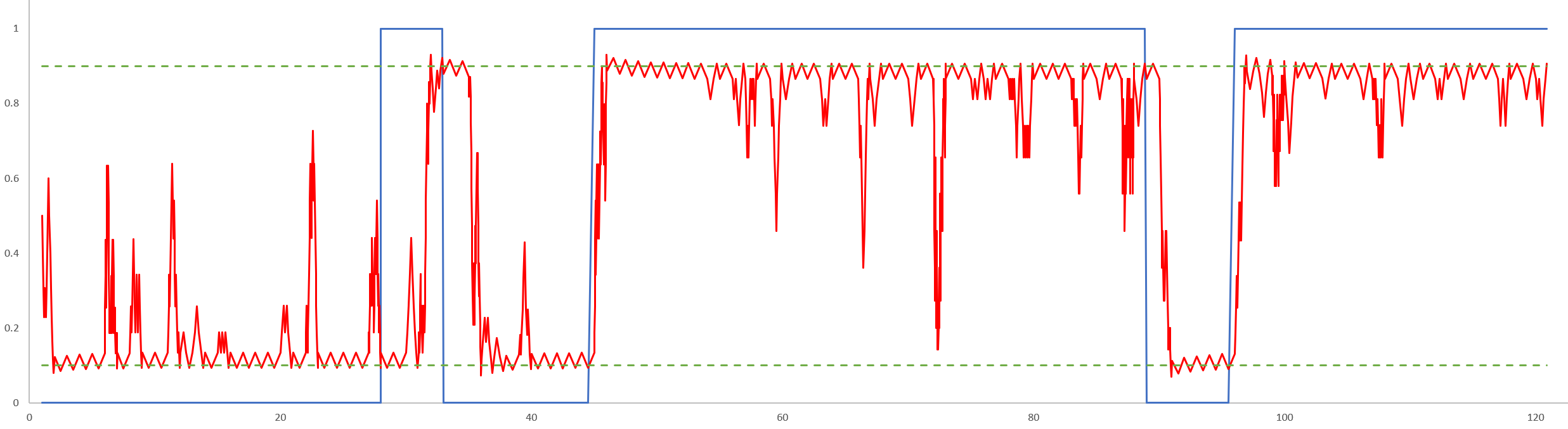}
\caption{Simulating the optimal policy for the non-Gaussian extension.  The round number is on the horizontal axis.  The hidden state of the world is binary and evolves stochastically (blue).  The optimal policy tracks a posterior distribution over the hidden state (red), and takes samples in order to maintain a tuned level of certainty (dashed green).  Note that most rounds have only a small number of samples, with occasional spikes triggered adaptively in response to uncertainty.}
\label{fig:extension}
\end{figure*}

In this extension, as in our Gaussian model, the optimal policy collects samples non-uniformly.  In fact, the optimal policy has a simple form: it sets a threshold $\theta > 0$ and takes samples until the entropy of the posterior distribution falls below $\theta$. 
Smaller $\theta$ leads to higher accuracy, but also requires more samples on average, so the best policy will set $\theta$ as low as possible subject to the budget constraint.  Notably, the result of this policy is that sampling tends to occur at a slow but steady rate, keeping the entropy around $\theta$, except for occasional spikes of samples in response to a perceived change in the hidden state.  See Figure~\ref{fig:extension} for a visualization of a numerical simulation with a budget of $6$ samples (on average) per round.

More generally, whenever the state evolves in a heavy-tailed manner, it is tempting to take samples regularly in order to detect large, infrequent
jumps in state value, and then adaptively take many samples when such a jump is evident.  This simple model is one scenario where such behavior is optimal.  More generally, can we quantify
the dynamic value of data and find an (approximately) optimal data collection policy for more complex Markov chains, or other practical applications?


\bibliographystyle{plain}
\bibliography{TemporalData}


\appendix

\section{Appendix: A Continuous Model}
\label{sec:model.continuous}

We now define a continuous version of our optimization problem, which is useful for modeling big-data situations in which the value of individual data samples is small, but the budget is large enough to allow the accumulation of large datasets.  Our continuous model will correspond to a limit of discrete models as the variance of the sampling errors and the budget $B$ grow large.

Our first step is to consider relaxing the discrete model to allow fractional samples.  This extends Lemma~\ref{lem:kalman} so that $s$ can be fractional.  Note that since the effect of taking $s$ samples, from Lemma~\ref{lem:kalman}, depends on the ratio $s/\sigma$, we can think of taking a fraction $\alpha < 1$ of a sample with variance $\sigma$ as equivalent to taking a single sample with variance $\sigma / \alpha$.  With this equivalence in mind, we can without loss of generality scale the variance of samples so that $\sigma = 1$; this requires only that we interpret the sample budget and numbers of samples taken as scaled in units of inverse variance.

Next, in the continuous model, the hidden state $x_t$ evolves continuously for $t \geq 0$.  The initial prior is Gaussian with some fixed variance $v_0$.  At each time $t$ we will have a posterior distribution over the hidden state.  Write $v(t)$ for the variance of the posterior distribution at time $t \geq 0$.  In particular, we have the initial condition $v(0) = v_0$.

Samples can be collected continuously over time at a specified density, as well as in atoms at discrete points of time.  As discussed above, we assume without loss of generality that the variance of a single sample is equal to $1$.
Write $s(t)$ for the density at which samples are extracted at time $t$, and write $a(t)$ for the mass of samples collected as an atom at time $t$.  Assume atoms are collected at times $\{t_i,i\in\mathbb{N}\}$, i.e., $a(t)>0$ only if $t\in\{t_i\}$. Both $s$ and $a$, as well as the times $\{t_i\}$ are chosen by the decision-maker.  

To derive the evolution of the hidden state and the variance $v(t)$ of the posterior at time $t$, we interpret this continuous model as a limit of the following discretization.  Partition time into intervals of length $\epsilon$, say $[t,t+\epsilon)$ for each $t$ a multiple of $\epsilon$.  We will consider a discrete problem instance, with discrete rounds corresponding to times $\epsilon, 2\epsilon, 3\epsilon, \dotsc$.  At round $i$, corresponding to time $t = i \cdot \epsilon$, a zero-mean Gaussian with variance $\epsilon$ is added to the state.  I.e., we take $\rho = \epsilon$ in our discrete model.  
%
We then imagine drawing $\int_{t = (i-1)\epsilon}^{i\epsilon} s(t)dt + \sum_{j:t_j \in [(i-1)\epsilon, i\epsilon)}a(t_j)$ samples at round $i$, corresponding to time $t = i\cdot \epsilon$. This represents the samples that would have been drawn over the course of the interval $[(i-1)\epsilon, i\epsilon)$.  We will also take the budget in this discrete approximation to be $B/\epsilon$, so that the approximation is valid if the continuous policy satisfies its budget requirement.  Note that we can think of this discretization as an approximation to the continuous problem, with the same budget $B$, but with time scaled by a factor of $\epsilon$ so that a single round in the discrete model corresponds to a time interval of length $\epsilon$ in the continuous model. 

Approximate $s$ by a function that is constant over intervals of length $\epsilon$, say equal to $s(t)$ over the time interval $[t, t+\epsilon)$ for each $t$ a multiple of $\epsilon$.  Suppose for now that there are no atoms over this period.
We are then drawing $\epsilon s(t)$ samples at time $t+\epsilon$.
By Lemma~\ref{lem:kalman}, this causes the variance to drop by a factor of $1+\epsilon s(t) v(t)$.
The new variance at time $t+\epsilon$ is therefore
\[ v(t+\epsilon)= \frac{v(t)+\epsilon}{1+\epsilon s(t) (v(t) + \epsilon)} 
\]
As this change occurs over a time window of length $\epsilon$, the average rate of change of $v$ over this interval is
\[\frac{v(t+\epsilon)-v(t)}{\epsilon}=\frac{1}{\epsilon}\left(\frac{v(t)+\epsilon}{1+\epsilon s(t) (v(t)+\epsilon)}-v(t)\right)
=\frac{1-s(t)v^2(t)-\epsilon s(t)v(t)}{1+\epsilon s(t)(v(t)+\epsilon)}\]
Taking the limit as $\epsilon \to 0$, the instantaneous rate of change of $v$ at $t$ is
\[ v'(t) = 1 - s(t) \cdot v^2(t). \]
The variance function $v$ is therefore described by the differential equation above, for any $t$ at which $a(t) = 0$.

If there is an atom at $t$, so that $a(t) > 0$, then for sufficiently small $\epsilon$ the number of samples in the range $[t, t+\epsilon)$ is instead $a(t) + \epsilon s(t)$.   As $\epsilon \to 0$, this introduces a discontinuity in $v(\cdot)$ at $t$, since the number of samples taken does not vanish in the limit.  With this in mind, we will take the convention that $v(t)$ represents $\lim_{t' \to^+ t}v(t)$, the right-limit of $v$; this informally corresponds to the variance ``after'' having taken atoms at $t$.  We then define $\tilde{v}(t) = \lim_{t' \to^- t}v(t)$ to be the variance ``before'' applying any such atom.  Lemma~\ref{lem:kalman} then yields that
\[ v(t) = \frac{\tilde{v}(t)}{1+a(t) \tilde{v}(t)}. \]
We emphasize that, under this notation, $v(t)$ represents the variance after having applied the atom at time $t$, if any.  These discontinuities combined with the differential equation above provide a full characterization of the evolution of the variance $v(\cdot)$, given $s(\cdot)$ and $a(\cdot)$ and the initial condition $v(0) = v_0$.

Then the total number of samples acquired over a time period $[0,T]$ is 
$\int_0^Ts(t)dt+\sum_{i:t_i\in[0,T]}a(t_i)$.
We normalize the cost per sample to 1 as in the discrete case, but modeling the fixed costs is more subtle.  In particular, consider some intermediate discretization.  If we apply the fixed cost $f$ at each interval, we get the counterintuitive result that taking $s$ samples today is less expensive than taking $s/2$ samples in the morning and $s/2$ in the afternoon, when logically the two should have at least similar costs.  On the other hand if we scale the cost to be $f/2$ we could only ever take samples in the morning and implement the same policy as we would have at the ``day'' level at half the fixed cost.  To avoid this we make the fixed cost history dependent.  If the fixed cost was not paid in the previous interval, the fixed cost is $f$.  If the fixed cost was paid in the previous interval, the cost is instead $\epsilon f$.  This ensures that the cost of implementing a policy from the original level of discretization at a finer level has the exact same cost, while keeping policies which spread out samples evenly over the interval at a similar cost.  Furthermore, to allow similar properties to hold when considering multiple possible levels of discretization, we allow the decision maker to pay the fixed cost even in periods when samples are not taken.  Thus, for example, taking samples in early morning and early afternoon but none in the late morning has the same cost as taking the same samples in the morning and in the afternoon.  This interpretation has natural analogs in some of our example scenarios, such as maintaining the satellite lock for the GPS even if limited numbers of samples are taken.

In the continuous limit, this cost model becomes a flow cost of $f$, which will be paid at all times samples are taken as well as during intervals when samples are not taken of length at most 1.  There is also a fixed cost $f$ when sampling resumes after an interval of length greater than 1.  Let $\phi$ be a measure which has density $f$ when when the flow cost is paid and measure $f$ at times when the fixed cost is paid.
Our budget requirement is that 
$\int_0^T s(t)dt+\sum_{i:t_i\in[0,T]}a(t_i) + \int_0^T d\phi \leq BT$ for all $T \geq 0$.  
%



The optimization problem in the continuous setting is to choose the functions $s(t)$ and $a(t)$ that minimizes 
the long-run average cost incurred by the decision-maker,
\[ \lim\sup_{T \to \infty}\frac{1}{T} \int_0^T \min(v(t), c)dt, \]
subject to the budget constraint.

Similar to the discrete setting, we define the \emph{null policy} to be the sampling policy that takes no samples (i.e., $s(t)=0$ and $a(t) = 0$ for all $t$) and selects the outside option at all times, for an average cost of $c$.  Again, we define the value of a policy to be the difference between its average cost and the average cost of the null policy:
\[ \lim\inf_{T \to \infty} \frac{1}{T} \int_0^{T} \max(c-v_t,0)dt. \]
We say a policy is $\alpha$-approximate if it achieves an $\alpha$ fraction of the value of the optimal policy.

\section{Appendix: Proofs from Section~\ref{sec:model.discrete}}
\label{app:model}

\newtheorem*{lem.kalman}{Lemma \ref{lem:kalman}}
\begin{lem.kalman}
Let $v_t$ be the variance of $G_t$ and suppose each $F_t$ is a zero-mean Gaussian with variance $\rho$, and that each sample is subject to zero-mean Gaussian noise with variance $\sigma$. Then, if the decision-maker takes $s$ samples in round $t+1$, the variance of $G_{t+1}$ is 
\[v_{t+1}=\frac{v_t+\rho}{1 + \frac{s}{\sigma}(v_t+\rho)}.\]
\end{lem.kalman}
\begin{proof}
Our model is a special case of the model underlying a Kalman filter.  There, generally, the evolution of the state can depend on a linear transformation of $x_t$, a control input, and some Gaussian noise.  In our model the transformation of $x_t$ is the identity, there is no control input, and by assumption the Gaussian noise is mean 0 variance $\rho$. Similarly, our sampling model corresponds to the observation model assumed by a Kalman filter.

Therefore, using the standard update rules for a Kalman filter~\cite{roweis1999unifying}, the innovation variance at time $t+1$ (i.e., the variance of the posterior after $x_t$ is updated by $\delta\sim F_{t+1}$ but before observing the samples) is $\tilde{v}_{t+1} = v_t + \rho$.  (Alternatively this can be observed directly as we are summing two Gaussians.)  This matches the desired quantity for the case $s=0$, where no samples are taken. 

For $s = 1$, we can again apply the standard update rules for a Kalman filter to get a posterior variance of $\frac{\tilde{v}_{t+1}}{1 + \frac{1}{\sigma}\tilde{v}_{t+1}}$, as desired.  By induction, if the decision maker instead takes $s > 1$ samples, the posterior variance will instead be $\frac{\tilde{v}_{t+1}}{1 + \frac{s}{\sigma}\tilde{v}_{t+1}}$, as desired.
\end{proof}

\section{Appendix: Proofs from Sections~\ref{sec:piecewise} and~\ref{sec:optimal}}
\label{app:opt}

All of these results hold in both the discrete and continuous models with essentially the same proofs.  Therefore, we provide  a unified treatment of them for both cases.

\newtheorem*{lem.contraction}{Lemma \ref{lem.contraction}}
\begin{lem.contraction}
Fix a sampling policy $s$, and a time $R > 0$, in either the continuous or discrete setting, and suppose that $s$ takes a strictly positive number of samples in $(0,R]$.  Let $\Psi$ be the mapping defined as follows: supposing that $v_0 = x$ and $v$ is the variance function resulting from sampling policy $s$, set $\Psi(x) := v(R)$.  Then $\Psi$ is a contraction map over the non-negative reals, under the absolute value norm.
\end{lem.contraction}
\begin{proof}
We need to show $|\Psi(x) - \Psi(y)| \leq k \cdot |x - y|$ whenever $x > y$, where $k < 1$ is some constant that depends on $s$.  We'll prove this first for a discrete policy.  Write $v^x_t$ and $v^y_t$ for the variance at time $t$ with starting condition $x$ and $y$, respectively.  Take $v^x_0 = x$ and $v^y_0 = y$ for notational convenience.  We then have that $v^x_1 = \frac{x+\rho}{1+(s_1/\sigma)(x+\rho)}$ and $v^y_1 = \frac{y+\rho}{1+(s_1/\sigma)(y+\rho)}$ by Lemma~\ref{lem:kalman}.  We then have that $v^x_1 \geq v^y_1$, and moreover
\[ v^x_1 - v^y_1 = \frac{x-y}{(1+(s_1/\sigma)(x+\rho))(1+(s_1/\sigma)(y+\rho))} \leq x-y \]
and the inequality is strict if $s_1 > 0$.  We can therefore apply induction on the rounds in $(0,R]$, plus the assumption that at least one of these rounds has a positive number of samples, to conclude that $\Psi(x) - \Psi(y) < x - y$.  To bound the value of $k$, suppose that $t \geq 1$ is the first round in which a positive number of samples is taken.  Then we can find some sufficiently small $\epsilon > 0$ so that $v_{t-1}^x > \epsilon$ and $s_t/\sigma > \epsilon$.  Then we will have
\begin{align*}
\Psi(x) - \Psi(y) &\leq v_{t}^x - v_{t}^y \\
&= \frac{v_{t-1}^x - v_{t-1}^y}{(1+(s_1/\sigma)(v_{t-1}^x+\rho))(1+(s_1/\sigma)(v_{t-1}^y+\rho))} \\
&\leq \frac{x - y}{1+(s_1/\sigma)(v_{t-1}^x+\rho)}\\ 
&< \frac{x-y}{1+\epsilon^2}
\end{align*}
and hence we have $k \leq \frac{1}{1+\epsilon^2} < 1$ as required.

To extend to continuous policies, take $v^x(t)$ and $v^y(t)$ for the variances under the two respective start conditions, and note first that there must exist some sufficiently small $\epsilon$ such that $v^x(t) > \epsilon$ for all $t \in [0,R]$, 
and for which the total mass of samples taken over range $(0,R]$ is at least $\epsilon$.  Take any discretization of the range $[0,R]$, say into $r > 1$ rounds, and consider the corresponding discretization of the continuous policy, so that the sum of the number of samples taken over all discrete rounds in the interval $[0,R]$ is at least $\epsilon$.  As above, take $v_i^x$ and $v_i^y$ to be the variances resulting from these discretized policies after $i$ discrete rounds.
Say $s_i$ samples are taken at round $i$.  Then, considering each round in sequence and applying the same reasoning as in the discrete case above, we have that
\begin{align*}
v^x_r - v^y_r 
&\leq (x-y) \cdot \prod_{i = 1}^r \frac{1}{1 + s_i v^x_{i-1} }\\
&\leq (x-y) \cdot \prod_{i = 1}^r \frac{1}{1 + s_i \epsilon }\\
&\leq (x-y) \cdot \frac{1}{1+ \epsilon \sum_{i=1}^r s_i} \\
&\leq \frac{x-y}{1 + \epsilon^2}.
\end{align*}
Thus, for each such discretization, we have a contraction by a factor of at least $\frac{1}{1+\epsilon^2}$.  Taking a limit of such discretizations, we conclude that this holds in the continuous limit as well, so that $\Psi(x) - \Psi(y) < \frac{1}{1+\epsilon^2}(x-y)$ as required.
%
%
\end{proof}

It is well known that a contraction mapping has a unique fixed point, and repeated application will converge to that fixed point.  Since we can view the impact of the periodic sampling policy as repeated application of mapping $\Psi$ to the initial variance in order to obtain $v(0), v(R), v(2R), \dotsc$, we conclude that the variance will converge uniformly to a periodic function for which $v(t) = v(t+R)$.  Thus, for the purpose of evaluating long-run average cost of a periodic policy, it will be convenient (and equivalent) to replace the initial condition on $v$, $v(0) = v_0$, with a periodic boundary condition $v(0) = v(R)$, and then choose $s$ to minimize the average cost over a single period:
\[ \frac{1}{R} \int_0^R \min\{v(t), c\}dt, \]
subject to the budget constraint that, at any time $T\in(0,R]$, we have\footnote{Note that we omit $t_i = 0$ from the summation over atoms, to handle the edge case where there is an atom at time $T$, and hence at time $0$ as well, which should not be counted twice.} $\int_0^T s(t)dt + \sum_{i:t_i \in (0,T]}a(t_i) \leq BT$.

For the remainder of the proofs in this section, we will allow fractional sampling rates even in the discrete setting.  Recall that one can define an $\alpha \in (0,1)$ fraction of a sample to be one in which the variance is increased by a factor of $1/\alpha$.

\newtheorem*{lem:onoff.monotone}{Lemma \ref{lem:onoff.monotone}}
\begin{lem:onoff.monotone}
With fractional samples, for all $T > T'$, we have $\Val(s^T) > \Val(s^{T'})$.
\end{lem:onoff.monotone}
\begin{proof}
If $s$ is an on-off policy with period $T'$, then the policy that uses the same ``on'' sampling rate with a period of $T$ has weakly better average value.  This is because the variance of this policy is decreasing over the ``on'' period, so is lowest at the end of the period.  Thus the optimal on-off policy with period $T$ has better average value than $s$.
\end{proof}

We will write $V^* = \sup_{T \to \infty} \Val(s^T)$.  From the lemma above, we have that $V^* = \lim_{T \to \infty} \Val(s^T)$ as well.  The following lemma will be useful for analyzing non-periodic policies.

\begin{lemma}
\label{lem:approx.periodic}
For any policy $s$, there is a sequence of policies $\{s^T\}$ for $T = 1, 2, \dotsc$, where $s^T$ is periodic with period $T$, such that $\lim_{T \to \infty} Val(s^T) \geq Val(s)$.
\end{lemma}
\begin{proof}
(sketch)  Fix $T$, let $W^T$ be the average value of policy $s$ over rounds $[0,T]$, and let $s^T$ be a periodic policy that mimics policy $s$ over rounds $[0,T]$.  Note that $\lim_{T \to \infty} W^T = Val(s)$.  The difference between $W^T$ and $Val(s^T)$ is driven by the initial condition: by Lemma~\ref{lem.contraction}, $Val(s^T)$ is simply the average period value under the boundary condition $v_0 = v_T$, whereas $s$ may have some alternative initial condition (say $v_0 = v$).  If the initial condition for $s$ lies below that of $s^T$, we will modify policy $s^T$ into a new periodic policy, as follows.  First, we'll note the (constant) number of extra samples needed on round $1$ to reach variance $v$ from an initially unbounded variance.  Our modified policy will first wait the (constant) number of rounds (say $r$) needed to acquire this much budget, without sampling.  Then, starting at round $r$, it will bring the variance to $v$ (using this accumulated budget) and then simulate policy $s$ over rounds $[0,T-r]$.  Relative to $W^T$, this policy obtains all of the value except for that accumulated by policy $s$ over rounds $[T-r, T]$, which is at most a constant.  This policy's value therefore matches $W^T$ in the limit as $T$ grows large.
\end{proof}

The following lemma shows that if there are intervals during which one takes the outside option, then it is better to have them occur at the beginning of the range $[0,R]$.  The intuition is that it is cheaper to reduce the variance to $c$ from a large value once than to reduce from a small value to $c$ many times.

\begin{lemma}
\label{lem:atomic.optimal.form}
Consider a valid policy,
given by $s(\cdot)$ and $a(\cdot)$, and any time $R>0$.  Then there is another 
valid policy $s^*(\cdot)$, $a^*(\cdot)$ and time $r \in [0,R]$ such that $a^*(t) = s^*(t) = 0$ for all $t < r$, $v^*(t) \leq c$ for all $t \in [r,R]$, and the average value of $a^*(\cdot)$ up to time $R$ is at least the average value of the original policy up to time $R$.  This is true in both the discrete and continuous models (with fractional samples).
\end{lemma}

\begin{proof}
We will write the following proof in the continuous model, but we note that the same proof applies to the discrete model with just minor adjustments to the notation.  Let $s(\cdot)$ and $a(\cdot)$ be a 
sampling
policy, and suppose that it does not satisfy the conditions of $s^*$ and $a^*$ in the lemma statement.  Write $v(\cdot)$ for the resulting variance, and suppose that $t_1$ is the infimum of all times for which $v(t) < c$.
Note that we can assume that $s(t) = 0$ for any $t$ such that $v(t) > c$, without loss.

Suppose time interval $[t_2, t_3)$ is the earliest maximal interval following $t_1$ such that $v(t) \geq c$ for all $t \in [t_2, t_3)$.  In other words, $[t_2,t_3)$ is an interval during which the decision-maker would choose the outside option, and this interval occurs after some point at which the decision-maker has not chosen the outside option.  
Such an interval must exist, since we assumed that the given policy does not satisfy the conditions of the lemma. 

Our strategy will be to transform this policy $a(\cdot)$ into a different policy that is closer to satisfying the conditions of the lemma.  Roughly speaking, we will do this by ``shifting'' the interval $[t_2,t_3)$ so that it lies before $t_1$: we will push 
the sampling policy over the range $[t_1, t_2)$ forward $(t_3 - t_2)$ units of time.
This will result in a policy with one fewer intervals of time in which the variance lies above $c$.  And, as we will show, this policy has the same total value as the original and is valid.  We can apply this construction to each such interval to construct the policy $s^*$ and $a^*$ required by the lemma.


Let us more formally describe what we mean by shifting the interval $[t_2, t_3)$.  We have that $v(t_2) = c$ and $a(t_3) > 0$ from the definitions of $t_2$ and $t_3$.  Write $\delta = t_3 - t_2$.  Then we have $\tilde{v}(t_3) = c+\delta$.
Let $\gamma > 0$ be such that $\frac{c+\delta}{1+\gamma(c+\delta)} = c$.  That is, $\gamma$ is the size of atom such that, if $a(t_3) = \gamma$, then we would have $v(t_3) = c$.  Since in fact we have $v(t_3) \leq c$ by maximality of the interval, and since $s(t) = 0$ for all $t \in (t_2, t_3)$ by assumption, it must be that $a(t_3) \geq \gamma$.  On the other hand, let $\gamma' > 0$ be such that $\frac{\tilde{v}(t_1)+\delta}{1 + \gamma'(\tilde{v}(t_1)+\delta)} = \tilde{v}(t_1)$.  Since $\tilde{v}(t_1) \geq c$, we must have $\gamma' \leq \gamma$.

We are now ready to describe the shifted policy, given by $s^*$ and $a^*$.  We set $s^*(t) = a^*(t) = 0$ for all $t < t_1 + \delta$, $a^*(t_1+\delta) = a(t_1) + \gamma'$, $a^*(t) = a(t-\delta)$ and $s^*(t) = s(t-\delta)$ for all $t \in (t_1 + \delta, t_3)$, $a^*(t_3) = a(t_3)-\gamma$, and $a^*(t) = a(t)$ and $s^*(t) = s(t)$ for all $t > t_3$.  Roughly speaking, the new policy ``moves'' the interval $[t_2, t_3)$, where the variance lies above $c$, to occur before the sampling behavior that began at time $t_1$.  It also reduces the atom at $t_3$ and increases the atom at $t_1$ (if any); the amounts are chosen so that $v^*(t_3) = v(t_3)$, as we shall see.

We claim that $v^*(t) \geq c$ for all $t < t_1 + \delta$, $v^*(t) = v(t-\delta)$ for $t \in [t_1+\delta, t_3)$, and $v^*(t) = v(t)$ for $t \geq t_3$.  This will imply that the average value of policy $a^*$ is equal to that of policy $a$, since they differ only in that a portion of the variance curve lying below $c$ has been shifted by $\delta$.  That $v^*(t) \geq c$ for all $t < t_1 + \delta$ follows from the definition of $a^*$.  That $v^*(t_1+\delta) = v(t_1)$ follows because $\tilde{v}^*(t_1+\delta) = \tilde{v}(t_1) + \delta$, and $a^*(t_1+\delta)$ consists of an atom $\gamma'$ that shifts the variance from $\tilde{v}(t_1)+\delta$ to $\tilde{v}(t_1)$, plus another atom $a(t_1)$ that shifts the variance from $\tilde{v}(t_1)$ to $v(t_1)$.  Given that $v^*(t_1+\delta) = v(t_1)$, we also have $v^*(t+\delta) = v(t)$ for all $t \in [t_1, t_2)$, as the policy $a^*$ is simply $a$ shifted by $\delta$ within this range.  Finally, we have $\tilde{v}^*(t_3) = c$, and $a^*(t_3) = a(t_3) - \gamma$, which is precisely the size of atom needed to shift the variance from $c$ to $v(t_3)$.  So we have that $v^*(t) = v(t)$ for all $t \geq t_3$, as $a^*$ and $a$ coincide for all such $t$.

We conclude that $s^*$ and $a^*$ have the same average value as $s$ and $a$.  Also, $s^*$ and $a^*$ uses less total budget than $s$ and $a$, and shifts some usage of budget to later points in time, so the new policy is valid.  Finally, the new policy has least one fewer maximal interval in which the variance lies strictly above $c$.  By repeating this construction inductively, we obtain the policy required by the lemma.
\end{proof}

We are now ready to show that no policy that satisfies the budget constraint can achieve value greater than $V^*$.
This establishes our main first claim, that on-off policies are optimal when $z = 0$.

\newtheorem*{thm:opt}{Theorem \ref{thm:opt}}
\begin{thm:opt}
With fractional samples, the value of any valid policy $s$ is at most $V^*$.
\end{thm:opt}
\begin{proof}
We will prove this claim under the discrete model.  Taking the limit over ever-finer discretizations then establishes the result for the continuous model as well.

Choose some $T > 0$, and fix any policy $s$ that is periodic with period $T$.  We will show that the average value of $s$ is at most $\Val(s^T) + o(1)$, where the asymptotic notation is with respect to $T \to \infty$.  Taking the limit as $T \to \infty$ and applying the Lemma~\ref{lem:approx.periodic} above will then complete the result.

Our approach to showing that the average value is at most $\Val(s^T) + o(1)$ will be to convert $s$ into an on-off policy $s'$ of period $T$, without decreasing its average value.  Recall from Lemma~\ref{lem.contraction} (and the discussion following its proof) that the long-run average value of $s'$ is simply the average period value under the periodic boundary condition $v'_0 = v'_T$.  Moreover, when constructing $s'$, we can without loss of generality relax the validity condition to be that at most $BT$ samples are taken at any point over the interval $[0,T]$.  This is because we can strengthen any policy under this weaker condition to satisfy the original budget constraint by delaying the start of the policy for $T$ rounds without taking any samples, and only then starting policy $s'$.  This will have the same long-run average value as our relaxed policy, again by Lemma~\ref{lem.contraction}.

We now describe a sequence of operations to convert $s$ into an on-off policy.  First, by Lemma~\ref{lem:atomic.optimal.form}, we can assume that $s$ spends no samples in the range $[0, T')$ for some $T' < T$, then has $v_t \leq c$ for all $t \in [T', T]$.
We will show that, given any policy of this form, one can convert it into an on-off policy without degrading the total value over the range $[0,T]$ by more than a constant.  

We will apply a potential argument.  Given a policy with variances given by $v_t$, we will write $\bar{v}$ for the average variance in the range $[T', T]$.  That is,
\[ \bar{v} \colon = \frac{1}{T-T'+1} \sum_{T' \leq t \leq T}v_t.\]
Let $v^+ = \max_{T' \leq t \leq T} v_t$ and $v^- = \min_{T' \leq t \leq T} v_t$.  That is, $v^+$ is the maximum variance and $v^-$ is the minimum variance achieved during the interval where the variance lies below $c$.  Note that we must have $v^+ \geq \bar{v}$ and $v^- \leq \bar{v}$.  Let $\Psi$ be the total number of timesteps $t$ between $T'$ and $T$ in which the variance is equal to either $v^-$ or $v^+$, plus $1$ if $v^+ = v^- = \bar{v}$.  Then $\Psi$ is an integer lying between $2$ and $T - T' + 2$.  Also, $\Psi = T - T' + 2$ only if $v^+ = v^-$, which implies that all variances are precisely equal to $\bar{v}$.  We will show how to modify a policy $s$ (in which $v^+ > v^-$) into a new policy $s'$ so that $\Psi$ strictly increases, without changing the average policy value.

Write $A$ for the set of timesteps with variance equal to $v^+$, and $B$ for the set of timesteps with variance equal to $v^-$.  Say $|A| = a$ and $|B| = b$.  Note that $A \cap B = \emptyset$, since we are assuming $v^+ > v^-$.  We will update the sampling policy so that, roughly speaking, the variance of the timesteps in $A$ are all decreased, and the variance of the timesteps in $B$ are all increased, until either a new timestep becomes either a maximal or minimal point, or until all timesteps have variance $\bar{v}$.  More formally, our update is parameterized by some $\epsilon > 0$, and satisfies the following conditions:
\begin{itemize}
\item at all timesteps $t \in A$, the variance is reduced by $\epsilon / a$,
\item at all timesteps $t \in B$, the variance is increased by $\epsilon / b$,
\item at timesteps not in $A \cup B$, the variance is unchanged,
\item $\epsilon$ is maximal so that all elements of $A$ still have the maximum variance, and all elements of $B$ still have the minimum variance.
\end{itemize}
Note that by the final condition, after making this change, either there will be one more maximal timestep or one more minimal timestep, or else the minimum equals the maximum.  In either case, $\Psi$ will strictly increase.  Moreover, this update does not change the average variance of the policy.

It remains to show that we can implement this update without increasing the total spend of the policy.  To see this, consider updating just a single timestep $t \in A$.  The change involves adding samples at time $t$ to decrease the variance by $\epsilon$, then removing samples from time $t+1$ to offset the resulting decrease at that point.  To decrease variance from $v_{t}$ to $v_{t} - \epsilon/a$ requires an extra $\frac{\epsilon/a}{v_{t}(v_{t} - \epsilon/a)}$ samples.  The number of samples that can be saved in the subsequent round is the amount required to move the variance from $v_t + 1$ to $v_t + 1 - \epsilon$, which is $\frac{\epsilon/a}{(v_t + 1)(v_t + 1 - \epsilon/a)}$.  The net increase in samples is therefore
\[ \left( \frac{\epsilon/a}{v_{t}(v_{t} - \epsilon/a)} - \frac{\epsilon/a}{(v_{t}+1)(v_{t} + 1 - \epsilon/a)}\right) \]
Applying this operation to all timesteps in $A$ (of which there are $a$), and recalling that they all have variance $v^+$, we have that the total cost in samples is
\[ \left( \frac{\epsilon}{v^+(v^+ - \epsilon/a)} - \frac{\epsilon}{(v^+ + 1)(v^+ + 1 - \epsilon/a)}\right)
= \epsilon \cdot \frac{2v^+ + 1 - \epsilon/a}{v^+(v^+ - \epsilon/a)(v^+ + 1)(v^+ + 1 - \epsilon/a)}. \]
A similar calculation yields that the total number of samples saved by increasing the variance by $\epsilon / b$ for all timesteps in $B$ is
\[ \left( \frac{\epsilon}{v^-(v^- + \epsilon/b)} - \frac{\epsilon}{(v^+ + 1)(v^+ + 1 + \epsilon/b)}\right)
= \epsilon \cdot \frac{2v^- + 1 + \epsilon/b}{v^-(v^- + \epsilon/b)(v^- + 1)(v^- + 1 + \epsilon/b)}. \]
Recalling that $v^+ - \epsilon/a \geq v^- + \epsilon/b$ and $v^+ > v^-$, we have that
\begin{align*}
\epsilon \cdot \frac{2v^- + 1 + \epsilon/b}{v^-(v^- + \epsilon/b)(v^- + 1)(v^- + 1 + \epsilon/b)}
&= \frac{(v^- + 1) + (v^- + \epsilon/b)}{v^-(v^- + \epsilon/b)(v^- + 1)(v^- + 1 + \epsilon/b)} \\
&= \frac{(v^+ + 1)\tfrac{v^+ - \epsilon/a}{v^- + \epsilon/b} + (v^+ - \epsilon/a)\tfrac{v^+ + 1}{v^- + 1}}{v^-(v^+ - \epsilon/a)(v^+ + 1)(v^- + 1 + \epsilon/b)} \\
&> \frac{(v^+ + 1) + (v^+ - \epsilon/a)}{v^-(v^+ - \epsilon/a)(v^+ + 1)(v^- + 1 + \epsilon/b)} \\
&> \frac{(v^+ + 1) + (v^+ - \epsilon/a)}{v^+(v^+ - \epsilon/a)(v^+ + 1)(v^+ + 1 - \epsilon/b)}
\end{align*}
and hence the total number of samples saved is greater than the total number of samples spent in making this change.  Thus, the new policy also satisfies the average budget constraint.

Repeating this procedure, we conclude that we must eventually reach a state in which the resulting policy has constant variance equal to $\bar{v}$ in the range $[T', T]$.  
This policy has $s_t = 1/\sqrt{\bar{v}}$ for all $t \in (T', T]$, and possibly a larger number of samples at $s_T'$.  The on-off policy that sets $s'_t = 1/\sqrt{\bar{v}}$ for all $t \in [T', T]$ is therefore also valid.  Moreover, from our previous analysis, this on-off policy has average value within $o(1)$ of policy $s$.  We conclude that the value of $s$ is at most $\Val(s^T) + o(1)$, as required.
\end{proof}

\section{Appendix: Proofs from Section~\ref{sec:lazy_approx} in the Continuous Model}
\label{sec:appendix.approx}

It is technically more convenient to present these results for the continuous model first, as that allows us to avoid rounding issues.  Therefore we first prove these results for the continuous model and then in Appendix~\ref{sec:discrete.proofs} provide proofs for the discrete version for those that do not have a unified proof here.

We first prove a structural result about transforming policies.

\newtheorem*{lem:atoms}{Lemma \ref{lem:atoms}}
\begin{lem:atoms}
Fix any valid sampling policy (not necessarily lazy) with resulting variance function $v(\cdot)$, and any sequence of timestamps $t_1 < t_2 < \dotsc < t_\ell < \dotsc$.  Then there is a valid policy that spends samples only in atoms, with $\{t ~|~ a(t) > 0 \} \subseteq \{t_1, \dotsc, t_\ell, \dotsc\}$, resulting in a variance function $\breve{v}(\cdot)$ with $\breve{v}(t_i) \leq v(t_i)$ for all $i$.
\end{lem:atoms}
\begin{proof}
We will prove this result for the continuous model.  The proof for the discrete model follows similarly, and is given in Section~\ref{sec:discrete.proofs}.

For the continuous model, we'll first prove the result for the case of a single timestep $t_1$.  The result then follows by repeated application to each subsequent $t_i$ inductively.  Let $s(t)$ and $a_0(t)$ be the continuous sampling rate and atoms of the original policy, respectively, with resulting variance $v(t)$.  Assume first that  $s(t) = 0$ for all $t$ and that the atoms in the interval $[0, t_1]$ occur at times $0 < \tau_1 < \tau_2 < \dotsc < \tau_k = t_1$, with corresponding number of samples $a_0(\tau_1), \dotsc, a_0(\tau_k)$.  Note that the assumption $\tau_k = t_1$ is without loss, as we could set $a_0(\tau_k) = 0$.  If $k = 1$ then we are done, so assume $k \geq 2$.  We will show that the alternative policy which lumps together the first two atoms, given by $a_1$, where $a_1(\tau_1) = 0$ and $a_1(\tau_2) = a_0(\tau_1) + a_0(\tau_2)$, and $a_1(\tau_i) = a_0(\tau_i)$ for all $i > 2$, results in a variance function $v_1$ such that $v_1(\tau_i) \leq v(\tau_i)$ for all $i \geq 2$.  This will complete the claim, by repeated application to the first non-zero atom in the sequence.

Recall that $\tilde{v}(\tau_1) = v(0) + \tau_1$ is the the variance just prior to applying the atom at $\tau_1$.  Then we have that $v(\tau_1) = \frac{\tilde{v}(\tau_1)}{1+a_0(\tau_1)\tilde{v}(\tau_1)}$.  Similarly, $\tilde{v}(\tau_2) = v(\tau_1) + (\tau_2 - \tau_1)$, so
\[ \tilde{v}(\tau_2) = \frac{\tilde{v}(\tau_1)}{1+a_0(\tau_1)\tilde{v}(\tau_1)} + \tau_2 - \tau_1. \]
We then have that $v(\tau_2) = \tilde{v}(\tau_2) / (1 + a_0(\tau_2) \tilde{v}(\tau_2))$.

Alternatively, with $a_1$ we have $\tilde{v}_1(\tau_2) = v(0) + \tau_2$. If we let $\tilde{\tilde{v}}_1(\tau_2)$ denote the variance after having applied an atom of size $a_0(\tau_1)$ at time $\tau_2$ but before applying an atom of size $a_0(\tau_2)$, we have
\[ \tilde{\tilde{v}}_1(\tau_2) = \frac{\tilde{v}(\tau_1) + \tau_2 - \tau_1}{1 + a_0(\tau_1)(\tilde{v}(\tau_1) + \tau_2 - \tau_1)}. \]
We will then have that $v_1(\tau_2) = \tilde{\tilde{v}}_1(\tau_2) / (1 + a_0(\tau_2) \tilde{\tilde{v}}_1(\tau_2))$.  We now note that $\tilde{\tilde{v}}_1(\tau_2) \leq \tilde{v}(\tau_2)$, since
\begin{align*}
\tilde{v}(\tau_2) 
&= \frac{\tilde{v}(\tau_1)}{1+a_0(\tau_1)\tilde{v}(\tau_1)} + \tau_2 - \tau_1 \\
&\geq \frac{\tilde{v}(\tau_1) + \tau_2 - \tau_1}{1+a_0(\tau_1)\tilde{v}(\tau_1)} \\
& \geq \frac{\tilde{v}(\tau_1) + \tau_2 - \tau_1}{1+a_0(\tau_1)(\tilde{v}(\tau_1) +\tau_2 - \tau_1)} \\
&= \tilde{\tilde{v}}_1(\tau_2).
\end{align*}
We can therefore conclude that $v(\tau_2) \geq v_1(\tau_2)$.  This further implies that $v(\tau_i) \geq v_1(\tau_i)$ for all $i > 2$ as well, since $a_1(\tau_i) = a(\tau_i)$ for all $i > 2$ and, inductively, the variance is weakly lower under $v_1$ than under $v$ just prior to each atom, and hence is lower after the application of each atom as well.

We now turn to the more general case where $s(t)$ is not identically $0$.  We again consider the case of a single timestep $t_1$, and the more general result will follow inductively.  
We can view $s$ as the limit, as $\epsilon \to 0$, of a sequence of discretized policies that only use atoms, and only at times that are multiples of $\epsilon$.  We will take our sequence to be $\epsilon = t_1/k$ for $k = 1, 2, 3, \dotsc$, so that time $t_1$ is present in each of these discretizations.  For each such $\epsilon = t_1/k$, take $s_k$ to be the discretized version of $s$, so that $s_k \to s$ as $k \to \infty$.  Applying our lazy result above to policy $s_k$, we have that for each $s_k$, there is a policy that matches the variance of $s_k$ at time $t_1$, and that only applies an atom at $t_1$.  Taking the limit as $k$ grows, we conclude that there is a policy that only takes an atom at time $t_1$, whose variance at $t_1$ is no greater than $v(t_1)$, the variance generated by policy $s$.
\end{proof}

We can now prove our main approximation result for the continuous setting without costs.  We actually prove two versions, as a stronger bound is possible for the case of $z=0$.

\newtheorem*{thm:2approx.costs1}{Theorem \ref{thm:2approx.costs} (Version 1)}
\begin{thm:2approx.costs1}
\label{thm:2approx.costs.continuous}
The optimal lazy policy is $1/2$-approximate, in the continuous setting and with $z = 0$.
\end{thm:2approx.costs1}
\begin{proof}
Write $s^*(t), a^*(t)$ for an optimal policy, and $v^*(t)$ for the corresponding variance function.  That is, $s^*$ is a policy that minimizes 
\[ \lim\sup_{T \to \infty} \frac{1}{T} \int_{t = 0}^{T} \min\{ c, v^*(t) \}dt \] 
subject to budget constraints.  

We note that, without loss of generality, we can assume that $s^*(t) = 0$ whenever $v^*(t) > c$.  That is, the policy has no continuous sampling when the variance is above the outside option.  This follows from our structural lemma above, since any policy can be replaced by one that takes no sampling in an interval where the variance lies above the outside option, and instead uses an atom at the end of such an interval to bring the variance back down to the outside option level $c$.

We now define an lazy policy that approximates $s^*,a^*$ up to a fixed $\epsilon > 0$.  We do so by defining a sequence of intervals iteratively.  We begin by setting $t_0 = 0$.  For each $t_i$, we will define $t_{i+1} > t_i$ as follows.  If $v^*(t) > c - \epsilon$ for all $t \in [t_i, t_i + \epsilon]$, take 
\[ t_{i+1} = \inf\{t > t_i\ \colon\ v^*(t) \leq c - \epsilon\}.\]  
Otherwise, choose $t_{i+1} = t_i + \delta$, where 
\[ \delta = \inf\{ \delta > \epsilon\ \colon\ c - v^*(t) \leq \delta\ \forall\ t \in [t_i, t_i+\delta) \}. \]
Note that in this latter case we must always have $\delta \geq \epsilon$, and hence $t_{i+1} \geq t_i + \epsilon$.  We also must have $\delta \leq c$, since certainly $v^*(t) \geq 0$ everywhere.

For each $i \geq 0$, let $m_i = \arg\inf_{m \in (t_i, t_{i+1}]}\{ v^*(m) \}$.  That is, $m$ is the time in the subinterval $(t_i, t_{i+1}]$ where $v^*$ takes its lowest value.  

Consider the policy that applies atoms at times $t_0, t_1, \dotsc$ and at times $m_0, m_1, \dotsc$, so as to match the variance of $v^*$ at each of those times.  By Lemma~\ref{lem:atoms}, this policy is valid.  

Next consider the policy that applies atoms only at times $t_0, t_1, \dotsc$, and applies those atoms so that $v(t_i) = v^*(m_i)$ for each $i$.  This policy is not necessarily valid, but we claim that this policy can only ever go budget negative by at most $c\cdot B$, the amount of budget accrued in $c$ time units.  This is because within each interval $[t_i, t_{i+1}]$, this policy uses no more budget than the previous policy.  It may cause budget to be spent earlier than before, within the same window. 
However, it only differs from the previous policy on subintervals where $v^*(m_i) < c - \epsilon$, and hence can only shift the spending of budget earlier by at most $c$ time units as, in these cases, $\delta$ (the difference between $t_i$ and $t_{i+1}$) is at most $c$.  Thus, this new policy can go budget-negative, but never by more than $cB$.  

We next claim this policy is lazy.  Indeed, for each sub-interval $[t_i, t_{i+1}]$, we have that $v(t_i) \geq c - \delta$ where $\delta = t_{i+1} - t_i$.  Thus, our policy has the property that $\lim_{t \to t_{i+1}}v(t) \geq c$, so each atom occurs at a point where the variance is at or above $c$.  Note that this makes use of the fact that variance drifts upward at a rate of $1$ per unit time, in the continuous model.

We claim this (budget-infeasible) policy is $1/2$-approximate, up to an additive $\epsilon$ term on the average cost.  Since the policy is lazy, its value (relative to the outside option) is the area of a sequence of isoceles right-angled triangles. By construction, the squares that form the completion of those triangles cover the entirety of the value of $s^*,a^*$, except possibly for regions where $v^*(t) \geq c - \epsilon$.  See Figure~\ref{fig:squares} for a visualization.
So our policy is $1/2$-approximate, possibly excluding regions where $v^*$ has an average contribution of $\epsilon$.

Finally, we note we can transform this policy to a budget-feasible one without loss in the approximation factor. First, if we shift our policy to start at time $c$, rather than time $0$, then it is precisely valid.  As this decreases its value by at most a bounded amount, the loss in average value over a time horizon $T$ vanishes as $T \to \infty$.  So we have a valid $1/2$-approximation subject to an arbitrarily small additional additive loss.  Taking $\epsilon \to 0$, and noting that there is a universally optimal lazy policy, gives us that the optimal lazy policy is exactly a $1/2$-approximation.
\end{proof}

\newtheorem*{thm:2approx.costs2}{Theorem \ref{thm:2approx.costs} (Version 2)}
\begin{thm:2approx.costs2}
The optimal lazy policy is $1/3$-approximate, in the continuous setting and with $z > 0$.
\end{thm:2approx.costs2}
\begin{proof}
We consider the same construction as in Theorem~\ref{thm:2approx.costs.continuous}.  The only step that could increase costs is the step that shifts atoms to the ``left endpoint'' of its corresponding square.  This might change the distance between atoms, possibly increasing costs by requiring us to pay the fixed startup cost more times.

To fix this, we'll change the definition of a square so that a new square cannot begin until the original policy takes a sample.  This might ``extend'' some squares to the right, forming rectangles.  The value generated from the extended part of the rectangle can be at most half the area of the square, since by definition the original policy isn't sampling during this time so its variance rises at rate 1.  This change therefore increases the approximation factor to at most 3, since now one ``triangle'' is covering a square plus one extra triangle.

Having made this change, we know by definition that the original policy sampled at the left endpoint of each rectangle.  We can therefore sample (only) at the left endpoint of each rectangle, without increasing costs relative to the original policy.  Note that this \emph{can} increase costs relative to the policy that takes atoms at the minimum-variance point within each rectangle; so our cost comparison will only be with respect to the original policy.
\end{proof}

We show how to extend these results to the discrete setting in Appendix~\ref{sec:discrete.proofs}.  We note that unlike the continuous setting, we does not suffer a loss in approximation factor for the discrete setting when adding fixed costs.

We next show how to compute the best lazy policy, a result which was sketched in Section~\ref{sec:lazy_approx}, but not formally stated.


\begin{lemma}
\label{lem:solve-optimal-atomic-continuous.costs}
One can compute an asymptotically optimal valid lazy policy in closed form.  
\end{lemma}
\begin{proof}
We will consider some large $R$ and solve for a policy that maximizes average value up to time $R$, then take a limit as $R \to \infty$.  By Lemma~\ref{lem:atomic.optimal.form}, we can assume the optimal policy sets $a(t) = 0$ for all $t < r$, and has $v(t) \leq c$ for all $t \geq r$, where $r \leq R$.

For a given atom of size $s$, say taken at time $t \geq r$, recall that the subsequent atom will be taken at time $t + (c - \tfrac{c}{1+sc})$, the next time at which the variance is equal to $c$.  We will therefore define the cost of this atom as $s + \min\{f, f(c - \tfrac{c}{1+sc})\}$.  This is the cost of the $s$ samples, plus the cost of either maintaining flow until the next atom is taken (if $c - \tfrac{c}{1+sc} < 1$) or of paying the start-up cost when the next atom is taken (otherwise).  The sum of costs over all atoms is equal to the total cost of the policy.  We therefore define the value density of the atom to be
\[ \frac{1}{2} \cdot \frac{1}{s + \min\{f, f(c - \tfrac{c}{1+sc})} \cdot \left(c - \frac{c}{1+sc}\right)^2. \]
This expression can be maximized with respect to $s$ by considering separately the cases $(c - \tfrac{c}{1+sc}) > 1$ and $(c - \tfrac{c}{1+sc}) \leq 1$.  The resulting solution will be the optimal choice of atom size, assuming that it does exhaust the total budget  I.e., when the total cost of the resulting policy is at least $BR$.

If the optimal value of $s$ corresponds to a policy that does not exhaust the total budget, then this means that it is time, rather than budget, that is the binding constraint.  In this case the optimal policy takes $r = 0$ and chooses $s$ so that the budget is exhausted at time $R$.  That is, $s$ is chosen so that $s + \min\{f, f(c - \tfrac{c}{1+sc})\} = B(c - \tfrac{c}{1+sc})$, meaning that the total cost of an atom equals the budget acquired over the time interval between that atom and the next.  Again, one can solve for $s$ by considering cases $(c - \tfrac{c}{1+sc}) > 1$ and $(c - \tfrac{c}{1+sc}) \leq 1$.  This $s$ will correspond to the optimal sampling policy, which will take atoms at regular intervals up to time $R$.
\end{proof}

\section{Appendix: Proofs from Section~\ref{sec:lazy_approx} in the discrete model}
\label{sec:discrete.proofs}


In this section we complete proofs of statements that hold in both the discrete and continuous models, that we have previously proved only in the continuous model.  We begin with the Lemma~\ref{lem:atoms}.

\begin{lem:atoms}
Fix any valid (not necessarily lazy) sampling policy with resulting variances $v_t$, and any sequence of timesteps $t_1 < t_2 < \dotsc < t_k < \dotsc$.  Then there is a valid policy that spends samples only at timesteps that lie in the set $\{t_1, \dotsc, t_k, \dotsc\}$, resulting in variances $\breve{v}_t$ with $\breve{v}_{t_i} \leq v_{t_i}$ for all $i$.
\end{lem:atoms}

\begin{proof}
We'll prove this for the interval $[0, t_1]$, and the result then follows by repeated application to each subsequent $t_i$ inductively.  Assume that the original policy and spends samples in the interval $(0, t_1]$ occur at times $0 < a_1 < a_2 < \dotsc < a_k = t_1$, with corresponding number of samples $s_1, \dotsc, s_k$, where $s_i \geq 0$ for each $i$.  Note that the assumption $a_k = t_1$ is without loss, as we could set $s_k = 0$.  If $k = 1$ then we are done, so assume $k \geq 2$.  We will show that the alternative policy with samples $s'_1, \dotsc, s'_k$, where $s'_1 = 0$ and $s'_2 = s_1 + s_2$, and $s'_i = s_i$ for all $i > 2$, results in a variance function $v'$ such that $v'(a_i) \leq v(a_i)$ for all $i \geq 2$.  This will complete the claim, by repeated application to the first non-zero atom in the sequence.

Write $v_1 = v(0) + a_1$ for the variance just prior to taking the samples at $a_1$.  Then we have that $v(a_1) = \frac{v_1}{1+s_1 v_1}$.  Writing $v_2 = v(a_1) + (a_2 - a_1)$ for the variance just prior to taking the samples at $a_2$, we have
\[ v_2 = \frac{v_1}{1+s_1v_1} + a_2 - a_1. \]
We then have that $v(a_2) = v_2 / (1 + s_2 v_2)$.

Alternatively, if we write $v'_2$ for the variance of the policy with $s'_1 = 0$, after having taken $s_1$ samples at time $a_2$ but before taking $s_2$ more samples,
\[ v'_2 = \frac{v_1 + a_2 - a_1}{1 + s_1(v_1 + a_2 - a_1)}. \]
We will then have that $v'(a_2) = v'_2 / (1 + s_2 v'_2)$.  We now note that $v'_2 \leq v_2$, since
\[ v_2 = \frac{v_1}{1+s_1v_1} + a_2 - a_1 \geq \frac{v_1 + a_2 - a_1}{1+s_1v_1}
\geq \frac{v_1 + a_2 - a_1}{1+s_1(v_1 + a_2 - a_1)} = v'_2. \]
We can therefore conclude that $v(a_2) \geq v'(a_2)$.  This further implies that $v(a_i) \geq v'(a_i)$ for all $i > 2$ as well, since $s'_i = s_i$ for all $i > 2$ and, inductively, the variance is weakly lower under $v'$ than under $v$ just prior to each set of samples, and hence is lower after the application of each set of samples.
\end{proof}

We next complete the proof of our main approximation result in the discrete setting. 

\newtheorem*{thm:2approx.costs}{Theorem \ref{thm:2approx.costs}}
\begin{thm:2approx.costs}
In the discrete model, the optimal lazy policy is $1/2$-approximate.
\end{thm:2approx.costs}
\begin{proof}
Write $s^*_t$ for the optimal sampling policy, and $v^*_t$ for the corresponding variances.  We note that, without loss of generality, we can assume that if $s^*_t >0$ then $v^*_t < c$.  That is, the policy does not take samples if the resulting variance is still above the outside option.  This follows from our structural lemma above, since any policy can be replaced by one that takes no samples until the variance would be below $c$, then takes all the forgone samples then.

We now define a lazy policy that approximates $s^*$.  We do so by defining a sequence of intervals iteratively.  We begin by setting $t_0 = 0$.  For each $t_i$, we will define $t_{i+1} > t_i$ as follows.  
Choose $t_{i+1} = t_i + \delta$, where 
\[ \delta = \inf\{ \delta \in \mathbb{N}^+ \colon\ c - v^*_t \leq \delta\ \forall\ t \in [t_i, t_i+\delta) \}. \]
We must have $\delta \leq c$, since certainly $v^*(t) \geq 0$ everywhere.

For each $i \geq 0$, let $m_i = \arg\inf_{m \in (t_i, t_{i+1}]}\{ v^*_m \}$.  That is, $m$ is the time in the subinterval $(t_i, t_{i+1}]$ where $v^*$ takes its lowest value.  

Consider the policy that takes samples at times $t_0, t_1, \dotsc$ and at times $m_0, m_1, \dotsc$, so as to match the variance of $v^*$ at each of those times.  By Lemma~\ref{lem:atoms}, this policy uses no more budget than $s^*$ at any given point of time, and it only takes samples on rounds when the original policy took samples, so is therefore valid even when $z > 0$.

Next consider the policy that applies atoms only at times $t_0, t_1, \dotsc$, and applies those atoms so that $v_{t_i} = v^*_{m_i}$ for each $i$.  This policy is not necessarily valid, but we claim that this policy can only ever go budget negative by at most $c\cdot B$, the amount of budget accrued in $c$ time units.  Proof of claim: within each interval $(t_i, t_{i+1}]$, this policy uses no more budget than the previous one.  It may cause budget to be spent earlier than before, within the same window.  However, it can only shift the spending of budget earlier by at most $c$ time units.  Thus, this new policy can go budget-negative, but never by more than $cB$.  

We next claim this policy is lazy.  Indeed, for each sub-interval $[t_i, t_{i+1}]$, we have that $v_{t_i} \geq c - \delta$ where $\delta = t_{i+1} - t_i$.  Thus, our policy has the property that $\lim_{t \to t_{i+1}^*}v(t) \geq c$, so each atom occurs at a point where the variance is at or above $c$.

Finally, we claim that our policy is $1/2$-approximate.  Since the policy is lazy, its value (relative to the outside option) is lower bounded by the area of a sequence of isoceles right-angled triangles. By construction, the squares that form the completion of those triangles cover the entirety of the value of $s^*$.

To complete the proof, we need to restore validity by shifting it to start at time $c$, rather than time $0$.  As this decreases its value by at most a bounded amount, the loss in average value over a time horizon $T$ vanishes as $T \to \infty$.
\end{proof}


\begin{lemma}
\label{lem:solve-optimal-atomic-discrete.costs}
One can compute a valid policy whose asymptotic value is at least the value of the optimal lazy policy.
\end{lemma}
\begin{proof}
We will consider some large $R$ and solve for a policy that maximizes average value up to time $R$, then take a limit as $R \to \infty$.  By Lemma~\ref{lem:atomic.optimal.form}, we can assume the optimal policy sets $s_t = 0$ for all $t < r$, and has $v(t) \leq c$ for all $t \geq r$, where $r \leq R$.  Suppose for now that we are allowed to take fractional samples.

Consider a lazy policy which takes samples until the variance is $v_r$ at time $r$, then every time it is due to take samples takes the number $s$ such that the variance becomes $r$ again.  By construction, this means $s$ samples are taken at each time other than $r$ at which they are taken,  which happens every $\delta = \lceil c - v_r \rceil$ periods.  Therefore,  $s$ satisfies the equation $v_r = \tfrac{v_r + \delta}{1 + s(v_r + \delta)}$, or $s = \tfrac{\delta}{v_r(v_r+\delta)}$.
We therefore define the value density of this policy as 
\begin{equation}
\label{eq:cost_form}
\frac{1}{\frac{\delta}{v_r(v_r+\delta)} + f} \cdot \left(\sum_{i = 0}^{\delta-1} c - v_r + i\right).
\end{equation}
We can then optimize over $v_r$ in the same manner as in Lemma~\ref{lem:solve-optimal-atomic-continuous.costs} to find the optimal way for a lazy policy to exhaust its budget.  For any given $R$ there may be a slight suboptimality due to the need to spend some samples to reach $v_r$ the first time and some budget that does not get spent if there is not budget for an integer number of spending periods, but these go to zero for $R$ large.

This gives an asymptotically valid lazy policy using fractional samples.  Note that the objective in Equation~\eqref{eq:cost_form} is quasiconcave in $v_r$ for each fixed $\delta$ (the derivative is a quadtratic in $v_r$ with negative coefficient on the $v^2$ term and positive on the other two terms).  Therefore the optimal integer choice for a given $\delta$ can be found by ``rounding'' $v_r$ up or down the smallest amount that gives an integer $s$ and choosing one of these.  While the resulting policy of waiting and taking $s$ samples every $\delta$ periods may not be lazy, it is asymptotically so for large $R$ by Lemma~\ref{lem.contraction}

\end{proof}

%

\section{Appendix: Optimal Policy for the Non-Gaussian Extension}
\label{sec:extension.proof}

In this section we solve for the form of the optimal policy in the binary extension discussed in Section~\ref{sec:extension}.

First we recall the model.  There is a binary hidden state of the world, $x_t \in \{0,1\}$, which flips each round independently with some small probability $\epsilon > 0$.  The decision-maker's action in each round is to guess the hidden state and the objective is to maximize the fraction of time that this guess is made correctly.  Write $y_t \in \{0,1\}$ for the guess made in round $t$.  Each sample is a binary signal correlated with the hidden state, equal to $x_t$ with probability $\tfrac{1}{2} + \delta$ where $\delta > 0$.  The decision-maker can adaptively request samples in each round, subject to the budget constraint (of $B > 0$ samples per round on average), before making a guess.  Note that sampling is adaptive: the decision-maker can observe the outcome of one sample in a round before choosing whether to take the next.  While this adaptivity was unimportant in the Gaussian setting (since the sample outcomes were not payoff-relevant, only the induced variance), it is significant for non-Gaussian evolution.

After sampling in each round, the decision-maker has a posterior distribution $G_t$ over the state of the world.  We'll write $G_{tk}$ for the posterior after $k \geq 0$ samples have been taken in round $t$.  Note that $G_{tk}$ is fully described by the probability that $x_t = 1$, which we will denote by $p_{tk}$.  We claim that there is an optimal policy that sets a threshold $\theta \in (0,1/2)$, and after having taken $k \geq 0$ samples in round $t$, it takes another sample if and only if $p_{tk} \in [\theta, 1-\theta]$.  We call such policies \emph{threshold policies}.

\begin{theorem}
For any $\epsilon$ and $\delta$, there is an optimal threshold policy.
\end{theorem}
\begin{proof}[Proof (sketch).]
Fix some large $T$.  We will evaluate performance over the first $T$ rounds, and show optimality up to a loss that is vanishing as $T$ grows large.

We first note that it suffices to consider policies that are admissible subject to a budget constraint that binds in expectation.  To see this, take a policy that satisfies the budget constraint in expectation; we will construct a policy with the same asymptotic value that satisfies the budget constraint ex post.  To do so, we'll make two changes to the budget-in-expectation policy.  First, delay the policy's execution by $\Theta(\sqrt{T})$ rounds.  This has negligible impact on asymptotic performance, but begins the policy with a pool of funds to pull from.  Standard concentration bounds will then imply that it the policy will exceed its ex post budget exponentially rarely; whenever it does, simply pause execution by another $\Theta(\sqrt{T})$ rounds to recover the pool of funds and then continue.  The resulting policy satisfies the budget ex post, and the impact of such pauses on asymptotic performance will vanish in the limit as $T$ grows large.

Next, we claim that it suffices to consider policies whose action after having taken $k \geq 0$ samples in round $t$ depends only on $p_{tk}$.  That is, the actions are otherwise history independent.  This is because the optimal policy starting from a state with posterior $p_{tk}$ depends only on $p_{tk}$, and in particular does not depend on the number of samples that have been taken on round $t$ or any previous round (as the constraints bind only in expectation).  Thus, for any optimal policy that sometimes takes an additional sample when the posterior is $p_{tk}$ and sometimes does not, it would likewise be optimal to simply ignore the history and choose a decision independently at random, according to a distribution consistent with how frequently each choice is made when the posterior is $p_{tk}$ over the long-run execution of the policy.


We next claim that the long-run payoff of the optimal policy, starting in a round $t$ where the posterior begins at $p_{t0}$, is weakly increasing in $|p_{t0} - 1/2|$.  This follows from the fact that $p_{t0}$ being farther from $1/2$ corresponds to additional certainty about the hidden state.  For any $1/2 < p'_{t0} < p_{t0}$, any policy with posterior mean $p_{t0}$ could choose to ``forget'' information and behave as though its posterior is $p'_{t0}$, and achieve at least as high a payoff (in expectation) as a policy whose true posterior in round $t$ is $p'_{t0}$.

We next claim that the total long-run payoff of guessing after having taken $k$ samples is increasing and weakly concave in $p_{tk}$, for $p_{tk} \geq 1/2$.  (The case $p_{tk} < 1/2$ will follow similarly by symmetry.)  The fact that payoffs are increasing follows by backward induction: this is certainly true in the last round $T$, as the payoff in round $t$ is strictly increasing.  Then, for any $t < T$, the payoff in the subsequent threshold policy likewise depends only on the value of the posterior at the beginning of the round (as argued above), which is increasing in $p_{tk}$ and will be at least $1/2$ if $p_{tk} \geq 1/2$.  Concavity follows from the fact that in-round payoffs increase linearly in $p_{tk}$, but inter-round reversion to the mean is more pronounced for larger $p_{tk}$, so an increase in $p_{tk}$ has a sublinear effect on the payoffs from subsequent rounds.

Since the payoff function is increasing in $p_{tk}$ in each round, the optimal policy will be monotone: in each round, there will be a threshold above which a guess is made, and below which samples are taken.  This choice of thresholds will be made to optimize long-run payoff given the average budget constraint.  Since the payoffs are concave in $p_{tk}$, and this value function is identical across rounds, the optimal choice of thresholds will likewise be uniform across rounds.  By symmetry, if $\theta$ is chosen as the threshold for $p_{tk} > 1/2$, the threshold for $p_{tk} < 1/2$ will be $1-\theta$.
\end{proof}

\end{document}